\documentclass[10pt, a4paper,reqno]{amsart}

\usepackage{color} \definecolor{bleu_sombre}{rgb}{0,0,0.6}  \definecolor{rouge_sombre}{rgb}{0.8,0,0}\definecolor{vert_sombre}{rgb}{0,0.6,0}
\usepackage[plainpages=false,colorlinks,linkcolor=bleu_sombre,citecolor=rouge_sombre,urlcolor=vert_sombre,breaklinks]{hyperref}

\usepackage[english]{babel}

\usepackage{amsmath,amssymb,amsthm,graphicx,amsfonts,url,color,enumerate,dsfont,stmaryrd}
\theoremstyle{plain}
\newtheorem{theorem}{{Theorem}}[section]
\newtheorem*{theorem*}{{Theorem}}
\newtheorem{proposition}[theorem]{Proposition}
\newtheorem*{proposition*}{Proposition}
\newtheorem{corollary}[theorem]{Corollary}
\newtheorem*{corollary*}{Corollary}
\newtheorem{lemma}[theorem]{Lemma}
\newtheorem*{lemma*}{Lemma}

\theoremstyle{definition}
\newtheorem{definition}[theorem]{Definition}
\newtheorem*{definition*}{Definition}

\theoremstyle{remark}
\newtheorem{remark}[theorem]{Remark}

\makeatletter

\@addtoreset{equation}{section}  % Pour que le compteur des équations reparte à 0 en changeant de section
\makeatother

\renewcommand{\leq}{\leqslant}	\renewcommand{\geq}{\geqslant}
\renewcommand{\bar}[1]{\overline{#1}}
\renewcommand\over[2]{{\,\buildrel #1\over#2\,}}

\newcommand{\inv}{^{-1}}

\newcommand{\abs}[1]{\left\vert #1\right\vert}        % valeur absolue
\newcommand{\nr}[1]{\left\Vert #1\right\Vert}         % norme
\newcommand{\innp}[2]{\left< #1 , #2 \right>}         % produit scalaire (inner product)  

\newcommand{\Dom}{\Dc}			% Domaine d'un opérateur
\newcommand{\pppg}[1] {\left< #1 \right>} 	% <x> = \sqrt{1+x^2}
\newcommand{\symb} {\Sc}		% Espace de symboles

\newcommand{\singl}[1]{\left\{ #1 \right\}}		% Singleton --> en fait, n'importe quel ensemble
\newcommand{\Ii}[2]   {\llbracket #1,#2 \rrbracket}      %{\{#1,\dots,#2\}} 	% intervalle d'entiers.
\newcommand{\R}{\mathbb{R}}		\newcommand{\C}{\mathbb{C}}
\newcommand{\N}{\mathbb{N}}

\newcommand{\1}[1]{\mathds 1 _{#1}}

\newcommand{\st}{\,:\,}					% ``tel que'' dans la définition d'un ensemble 

\newcommand{\divg}{\mathop{\rm{div}}\nolimits}
\newcommand{\restr}[2]{\left.#1\right|_{#2}}         % #1 restreint à #2
\renewcommand{\Re}{\mathop{\rm{Re}}\nolimits}        % partie réelle
\renewcommand{\Im}{\mathop{\rm{Im}}\nolimits}        % partie imaginaire, Image
	
\DeclareMathOperator{\Id}{Id}                        % identité
 
\renewcommand{\a}{\alpha}\renewcommand{\b}{\beta}\newcommand{\g}{\gamma}\renewcommand{\d}{\delta}\newcommand{\D}{\Delta}\newcommand{\e}{\varepsilon} \newcommand{\y}{\eta}\renewcommand{\th}{\theta}\newcommand{\Th}{\Theta}\renewcommand{\l}{\lambda}\newcommand{\m}{\mu}\newcommand{\n}{\nu}\newcommand{\x}{\xi}\newcommand{\s}{\sigma}\renewcommand{\t}{\tau}\newcommand{\f}{\varphi}\newcommand{\vf}{\phi}\newcommand{\h}{\chi}\newcommand{\p}{\psi}\renewcommand{\o}{\omega}\renewcommand{\O}{\Omega}

\newcommand{\Ac}{{\mathcal A}}\newcommand{\Cc}{{\mathcal C}}\newcommand{\Dc}{{\mathcal D}}\newcommand{\Ec}{{\mathcal E}}\newcommand{\Hc}{{\mathcal H}}\newcommand{\Kc}{{\mathcal K}}\newcommand{\Lc}{{\mathcal L}}\newcommand{\Oc}{{\mathcal O}}\newcommand{\Rc}{{\mathcal R}}\newcommand{\Sc}{{\mathcal S}}\newcommand{\Tc}{{\mathcal T}}\newcommand{\Uc}{{\mathcal U}}\newcommand{\Vc}{{\mathcal V}}\newcommand{\Wc}{{\mathcal W}}

\newcommand{\ad}{{\rm{ad}}}

\newcommand{\detail}[1]
{
}

\newcommand{\qandq}{\quad \text{and} \quad}

\newcounter{stepproof}
\newcommand{\stepp}{\stepcounter{stepproof} \noindent {\bf $\bullet$}\quad }

\begin{document}

\newcommand{\Pg}{P}\newcommand{\Ho}{\Pg}
\newcommand{\Pgg}{P_G}
\newcommand{\Hz}{H_z}
\newcommand{\Pii}{P_\y}
\newcommand{\tPii}{\hat P_{\y,z}}
\newcommand{\Pcc}{P_{\y,c}}
\newcommand{\Hii}{H_{\y,z}}
\newcommand{\tHii}{\hat H_{\y,z}}
\newcommand{\Hcc}{H_{\y,c}}
\newcommand{\Ri}{R_{\y}}
\newcommand{\Riz}{\hat R_{\y}}
\newcommand{\RRiz}{\hat \Rc^{\y}}
\newcommand{\Kcc}{K_{\y}(z)}
\newcommand{\Kcco}{K_{\y}}

\newcommand{\tThiota}{\tilde \Th}
\newcommand{\Rgg}{R_{G}}
\newcommand{\Agg}{\Ac_G}

\newcommand{\nul}{{\nu_l}}
\newcommand{\nur}{{\nu_r}}
\newcommand{\nus}{{\nu_*}}
\newcommand{\anul}{{\abs \nul}}
\newcommand{\anur}{{\abs \nur}}
\newcommand{\anus}{{\abs \nus}}
\newcommand{\tnul}{{\tilde \nu_l}}
\newcommand{\tnur}{{\tilde \nu_r}}
\newcommand{\tnus}{{\tilde \nu_*}}
\newcommand{\mul}{{\mu_l}}
\newcommand{\mur}{{\mu_r}}
\newcommand{\Rci}{\Rc^\y}
\newcommand{\tdelta}{\s}
\newcommand{\Phil}{\Phi_l(z)}
\newcommand{\Phir}{\Phi_r(z)}
\newcommand{\tPhil}{\tilde \Phi_l(z)}
\newcommand{\tPhir}{\tilde \Phi_r(z)}
\newcommand{\tVc}{\tilde \Vc}

\newcommand{\sqP}{\sqrt{\Pg}}

\title{Local decay for the damped wave equation in the energy space}
\author{Julien Royer}
\curraddr{Institut de Math\'ematiques de Toulouse \\
118 route de Narbonne, 31062 Toulouse C\'edex 9}
\email{julien.royer@math.univ-toulouse.fr}


\begin{abstract}
We improve a previous result about the local energy decay for the damped wave equation on $\R^d$. The problem is governed by a Laplacian associated with a long range perturbation of the flat metric and a short range absorption index. Our purpose is to recover the decay $\Oc(t^{-d+\e})$ in the weighted energy spaces. The proof is based on uniform resolvent estimates, given by an improved version of the dissipative Mourre theory. In particular we have to prove the limiting absorption principle for the powers of the resolvent with inserted weights.
\end{abstract}

\maketitle

\section{Introduction and statements of the results}

We consider on $\R^d$, $d\geq 3$, the damped wave equation:
\begin{equation} \label{wave-lap}
 \begin{cases}
\partial_t^2 u  + \Ho u + a(x) \partial_t u = 0,  & \text{on  }  \R_+ \times \R^d, \\
\restr{(u , \partial_t u )}{t = 0} = (u_0, u_1), &  \text{on } \R^d.
 \end{cases}
\end{equation}
The operator $\Pg$ is a Laplace-Beltrami operator (or a Laplacian in divergence form) associated to a long-range perturbation of the usual flat metric (see \eqref{hyp-long-range} below). In particular it is a self-adjoint and non-negative operator on some weighted space $L^2_w = L^2(w(x) \, dx)$ where $w$ is bounded above and below by a positive constant ($w = \det (g(x))^{\frac 12}$ if $\Pg = -\D_g$ and $w = 1$ for a Laplacian in divergence form). The absorption index $a$ is smooth, takes non-negative values and is of short range (see \eqref{hyp-a-short-range}). Our purpose is to improve the result of \cite{boucletr14} concerning the local energy decay for the solution of this problem.\\

Consider $u_0,u_1$ in the Schwartz space $\Sc$, and let $u$ be the solution of \eqref{wave-lap}. It is straightforward to check that if $a = 0$ the total energy 
\[
E_P(t) := \nr{\sqrt \Pg u(t)}^2_{L^2_w} + \nr{\partial_t u(t)}^2_{L^2_w}
\]
is conserved. In general we have 
\[
\frac {d} {dt}E_\Pg(t)= - \int_{\R^d} a(x) \abs{\partial_t u (t)}^2 w(x)\,dx \leq 0. 
\]

In \cite{boucletr14} we have proved that if all the bounded geodesics go through the damping region $\singl{a > 0}$ (this is the so-called Geometric Control Condition, see Assumption \eqref{hyp-damping} below), then the energy which is not dissipated by the medium eventually escapes to infinity, as is the case for the analogous self-adjoint problem under the usual non-trapping condition (see \eqref{hyp-non-trapping}). 

This work came after many papers dealing with the self-adjoint case $a=0$. We mention for instance \cite{laxmp63} for the free wave equation outside some star-shaped obstacle, \cite{morawetzrs77} and \cite{melrose79} for a non-trapping obstacle, \cite{ralston69} for the necessity of the non-trapping condition to obtain uniform local energy decay and \cite{burq98} for a logarithmic decay with loss of regularity but without any geometric assumption. All these papers deal with a self-adjoint and compactly supported perturbation of the free wave equation on the Euclidean space. We also mention \cite{bonyh12} and \cite{bouclet11} for a long-range perturbation of the free laplacian, and \cite{alouik02,khenissi03} for the dissipative wave equation outside a compact obstacle. In \cite{boucletr14} we considered the damped wave equation for a Laplacian associated to a long-range perturbation of the flat metric. Here we improve the result obtained in this setting.\\

Let $\Hc_\Pg$ be the Hilbert completion of $\Sc \times \Sc$ for the norm 
\[
\nr{(u,v)}_{\Hc_\Pg}^2 = \nr{\sqrt \Pg u}^2_{L^2_w} + \nr{v}^2_{L^2_w}.
\]
We consider on $\Hc_\Pg$ the operator 
\begin{equation} \label{def-Ac}
 \Ac = \begin{pmatrix} 0 & I \\ \Ho & -i a \end{pmatrix}
\end{equation}
with domain 
\begin{equation} \label{dom-A}
\Dom(\Ac) = \singl{(u,v) \in \Hc_\Pg \st (v , \Pg u) \in \Hc_\Pg}.
\end{equation}
Then $u$ is a solution to the problem \eqref{wave-lap} if and only if $U = (u,i \partial_t u)$ is a solution to
\begin{equation} \label{wave-A}
\begin{cases}
 (\partial_t  + i \Ac  ) U(t) = 0,\\
U(0) = U_0,
\end{cases}
\end{equation}
where $U_0 = (u_0,i u_1)$. The operator $\Ac$ is maximal dissipative on $\Hc_\Pg$ (see Proposition 3.5 in \cite{boucletr14}). This implies in particular that $-i\Ac$ generates a contractions semigroup, and hence for $U_0 \in \Dom(\Ac)$ the problem \eqref{wave-A} has a unique solution $U : t \mapsto  e^{-it\Ac} U_0 \in C^0(\R_+,\Dom(\Ac)) \cap C^1(\R_+^*, \Hc_\Pg)$. Then the first component $u$ of $U$ is the unique solution of \eqref{wave-lap}.
\\

Now we describe more precisely the operators $\Pg$ which we consider. We first consider the case of a Laplace-Beltrami operator associated to a metric $g(x)$:
\[
\Pg := - \D_g = - \sum_{j,k=1}^d \abs{g(x)}^{-\frac 12} \frac \partial {\partial x_j} \abs{g(x)}^{\frac 12} G_{j,k}(x) \frac {\partial} {\partial x_k},
\]
where $\abs{g(x)} = \det (g_{j,k}(x))$ and $(G_{j,k}) = (g_{j,k})\inv$. The metric $g(x)$ is a long-range perturbation of the flat metric: for some $\rho > 0$ we have 
\begin{equation} \label{hyp-long-range}
\abs{\partial^\a \big( g(x) - I_d \big)} \leq C_\a \pppg x ^{-\rho - \abs \a},
\end{equation}
where $\pppg x = \big( 1 + \abs x^2 \big) ^{\frac 12}$. The same holds for $G$. We recall from \cite{bouclet11} that we can assume without loss of generality that $\abs {g(x)} = 1$ outside some compact subset of $\R^d$. Thus there exist $b_1,\dots b_d \in C_0^\infty(\R^d)$ such that $-\D_g$ is of the form
\begin{equation} \label{def-Pg}
\Pg = - \divg G(x) \nabla + W, \quad \text{where }  W = \sum_{j=1}^d b_j(x) D_j.
\end{equation}
Here and everywhere below, $D_j$ stands for $-i\partial_{x_j}$. Now let $w = \abs{g(x)}^{\frac 12}$. Then $w$ is bounded above and below by positive constants, and 
\begin{equation} \label{Pg-selfadjoint}
\Pg \text{ is self-adjoint and non-negative on $L^2_w$ with domain $H^2_w$,}
\end{equation}
where $H^2_w$ is the usual weighted Sobolev space endowed with the norm $\nr{u}_{H^2_w} = \sum_{\abs \a \leq 2} \nr{w \partial^a u}_{L^2}$.
As suggested by \eqref{def-Pg}, we will see the Laplace-Beltrami operator as a perturbation of a Laplacian in divergence form. Thus we can also simply consider the case 
\[
\Pg := -\divg G(x) \nabla,
\]
where $G$ satisfies \eqref{hyp-long-range}. Then \eqref{def-Pg} and \eqref{Pg-selfadjoint} hold with $b_1 = \dots = b_d = 0$ (so $W = 0$) and $w = 1$.

Concerning the dissipative term, we said than the absorption index $a$ is non-negative and of short range. This means that
\begin{equation} \label{hyp-a-short-range}
\abs{\partial^\a a(x)} \leq C_\a \pppg x ^{- 1 -\rho - \abs \a}.
\end{equation}

For $\d \in \R$ we denote by $\Hc^\d$ the Hilbert completion of $\Sc \times \Sc$ for the norm 
\begin{equation} \label{norm-Hc-delta}
\nr{(u,v)}_{\Hc^\d} := \nr{\pppg x^\d \nabla u}_{L^2} + \nr{\pppg x^\d v}_{L^2}.
\end{equation}
We also set $\Hc = \Hc^0$. Then we have $\Hc = \dot H^1 \times L^2$, where $\dot H^1$ is the usual homogeneous Sobolev space. We observe that the spaces $\Hc$ and $\Hc_\Pg$ are equal with equivalent norms.\\

The main result of this paper is the following:

\begin{theorem}  [Local energy decay] \label{th-loc-decay-improved}
Assume that every bounded geodesic goes through the damping region (see \eqref{hyp-damping} below). Let $\d > d + \frac 12$ and $\e > 0$. Then there exists $C \geq 0$ such that for $U_0 \in \Hc^{\d}$ and $t \geq 0$ we have
\[
\nr{e^{-it\Ac} U_0}_{\Hc^{-\d}} \leq C \pppg t ^{-(d-\e)} \nr{U_0}_{\Hc^\d}.
\]
With the notation of \eqref{wave-lap}, this estimate reads
\[
\nr{\pppg x^{-\d} \nabla u(t)}_{L^2} + \nr{\pppg x^{-\d} \partial_t u(t)}_{L^2} \leq C \pppg t ^{-(d-\e)} \left( \nr{\pppg x^{\d} \nabla u_0}_{L^2} + \nr{\pppg x^{\d} u_1}_{L^2} \right).
\]
\end{theorem}

The proof of this result relies on uniform resolvent estimates for the operator $\Ac$. After a Fourier transform, the solution $U$ of \eqref{wave-A} can be written as the integral over $\t \in \R$ of $(\Ac-(\t+i0))\inv$. Thus we need estimates for the resolvent $(\Ac-z)\inv$ when $\Im(z) \searrow 0$. As usual, the difficulties arise when $\t$ is close to 0 and for $\abs \t \gg 1$. Let 
\[
\C_+ = \singl{z \in \C \st \Im(z) > 0}.
\]
We begin with the statement concerning the intermediate frequencies:

\begin{theorem}[Intermediate frequency estimates] \label{th-inter-freq}
Let $N \in \N$, $\d > N + \frac 12$ and $\g \in ]0, 1]$. Then there exists $C \geq 0$ such that for all $z \in \C_+$ with $\g \leq \abs z \leq \g\inv$ we have 
\[
\nr{(\Ac-z)^{-1-N}}_{\Lc(\Hc^\d,\Hc^{-\d})} \leq C.
\]
\end{theorem}

For high frequencies, we know that the wave propagates along the underlying classical flow (in a sense made rigorous by semiclassical analysis, see for instance \cite{zworski} for the general theory and the Egorov Theorem in particular). Here this corresponds to the geodesic flow on $\R^{2d} \simeq T^* \R^d$ for the metric $G(x)\inv$ (that is the geodesic flow of the metric $g(x) $ when $\Pg = -\D_g$). It is the Hamiltonian flow corresponding to the symbol 
\[
p(x,\x) = \innp{G(x) \x}{\x}.
\]
We denote by $\vf^t = (X(t),\Xi(t))$ this flow. Let 
\[
\O_b = \singl{w \in p\inv(\singl 1) \st \sup _{t\in\R} \abs {X(t,w)} < +\infty}
\]
be the set of bounded geodesics. We say that the classical flow is non-trapping if 
\begin{equation} \label{hyp-non-trapping}
\O_b = \emptyset.
\end{equation}
We say that the damping condition on bounded geodesics (or Geometric Control Condition, see \cite{raucht74,bardoslr92}) is satisfied if every bounded geodesic goes through the damping region $\singl{a(x) > 0}$:
\begin{equation} \label{hyp-damping}
\forall w \in \O_b, \exists T \in \R, \quad a \big(X(T,w)\big) > 0.
\end{equation}
In particular we recover the non-trapping condition when $a = 0$. Under this damping assumption we can prove the following result:

\begin{theorem}[High frequency estimates] \label{th-high-freq}
Let $N \in \N$, $\d > N + \frac 12$ and $\g > 0$. Assume that the damping condition \eqref{hyp-damping} is satisfied. Then there exists $C \geq 0$ such that for all $z \in \C_+$ with $\abs z \geq \g$ we have 
\[
\nr{(\Ac-z)^{-1-N}}_{\Lc(\Hc^\d,\Hc^{-\d})} \leq C.
\]
\end{theorem}

It is known that for low frequencies we cannot estimate uniformly all the derivatives of the resolvent. This explains the restriction of the rate of decay in Theorem \ref{th-loc-decay-improved}.

\begin{theorem}[Low frequency estimates] \label{th-low-freq}
Let $N \in \N$, $\d > N + \frac 12$ and $\e > 0$. Then there exist a neighborhood $\Uc$ of 0 in $\C$ and $C \geq 0$ such that for all $z \in \Uc \cap \C_+$ we have 
\[
\nr{(\Ac-z)^{-1-N}}_{\Lc(\Hc^\d,\Hc^{-\d})} \leq C \left( 1 + \abs z ^{d-1-N-\e} \right).
\]
\end{theorem}

In order to prove Theorems \ref{th-inter-freq}, \ref{th-high-freq} and \ref{th-low-freq} we estimate the resolvent for the corresponding Schr\"odinger operator on $L^2$. We recall from Propositions 3.4 and 3.5 in \cite{boucletr14} that for $z \in \C_+$ we have on $\Sc \times \Sc$
\begin{equation} \label{res-Ac}
(\Ac-z)\inv =
\begin{pmatrix} 
R(z) (ia + z) &   R(z)\\
I +  R(z) (zia + z^2) &  z R(z)
\end{pmatrix}
= 
\begin{pmatrix} 
R(z) (ia + z) &   R(z)\\
R(z) \Ho &  z R(z)
\end{pmatrix},
\end{equation}
where 
\begin{equation} \label{def-R}
R(z) = \big( \Pg - iz a(x) - z^2 \big) \inv.
\end{equation}
The resolvent estimates of $\Ac$ on $\Hc$ will be deduced from estimates for $R(z)$ on $L^2$.\\

The purpose of this paper is not to improve the rate of the local energy decay, which is the same as in \cite{boucletr14}. Even in the self-adjoint setting (see \cite{bonyh12}), which is contained in our result, this is the best decay known for a general long-range perturbation of the Laplacian. We recall that estimates of size $\Oc(t^{-d})$ have been obtained in \cite{guillarmouhs13} ($d$ odd) and \cite{tataru13} ($d = 3$), but the metric is of scattering type in the first case, and it is radial up to short range terms in the second. We also recall that a stronger damping would not improve the estimate. On the contrary, with a strong damping the contribution of low frequencies tends to behave as the solution of a heat equation, for which the local energy decay is weaker (see for instance \cite{marcatin03, nishihara03, narazaki04,hosonoo04}).

The difference with \cite{boucletr14} is that we have an estimate in the (weighted) energy space. The improvement is twofold. The main point is that we get rid of the loss of decay: if $U_0 = (u_0,iu_1)$, then the estimates no longer depend on the sizes of $u_0$ and $u_1$ in the (weighted) Sobolev spaces $H^{2,\d}$ and $H^{1,\d}$, respectively. For this we will use a method adapted from \cite{rauch78,tsutsumi84} to deduce the time decay from the resolvent estimates for the contribution of high frequencies.

The second point is that the estimate no longer depends on the size of $u_0$ in $L^{2,\d}$, as is the case in \cite{boucletr14} or \cite{kopylova10}. This means in particular that we cannot deduce directly the estimates for $(\Ac-z)\inv$ from the estimates of $R(z)$ as an operator on $L^2$. However this is natural since the energy of the wave does not depend on the $L^2$-norm of $u_0$. This question has not been raised in the above mentioned papers for the following two reasons. First, when $u_0$ is supported in a fixed compact of $\R^d$ (as is the case in most of the papers dealing with the local energy decay) then by the Poincar\'e inequality the difference between the norms of $u_0$ in the homogeneous or inhomogeneous Sobolev spaces $\dot H^1$ and $H^1$ is irrelevant.
On the other hand, in the self-adjoint case we can write 
\begin{equation} \label{diag-Ac0}
\Ac_0 := \begin{pmatrix} 0 & I \\ \Ho & 0 \end{pmatrix} = \Phi \inv \Tc_0 \Phi,
\end{equation}
where 
\begin{equation*} \label{def-P}
\Tc_0 =  \begin{pmatrix} \sqP & 0 \\ 0 & - \sqP \end{pmatrix}, \quad 
\Phi = \frac 1 {\sqrt 2} \begin{pmatrix}  \sqP & 1 \\  \sqP & - 1 \end{pmatrix}
\quad \text{and} \quad  
\Phi\inv = \frac 1 {\sqrt 2}  \begin{pmatrix} 1/\sqP & 1/\sqP \\ 1 & -1 \end{pmatrix}  .
\end{equation*}
The operators $\Phi$ and $\Phi\inv$ are isometries in $\Lc \big(\Hc_\Pg , (L^2_w)^2 \big)$ and $\Lc \big( (L^2_w)^2,\Hc_\Pg \big)$ respectively, and $\Tc_0$ defines on operator on $(L^2_w)^2$ with domain $\Dom(\sqP)^2$. Thus the properties of $\Ac_0$ on $\Hc_\Pg$ follow directly from the analogous properties of $\sqP$ on $L^2_w$. In particular the resolvent estimates for $\Ac_0$ follow from the estimates on $(\sqP - z)\inv$ and $(-\sqP-z)\inv$ in $L^2_w$. Back in the energy space, this gives estimates for the norms
\begin{equation} \label{norm-HcPg-delta}
\nr{(u,v)}^2_{\Hc_\Pg^\d} = \nr{\pppg x^\d \sqP u}_{L^2_w}^2 + \nr{\pppg x^\d v}^2_{L^2_w}.
\end{equation}
This is what is implicitely used for instance in \cite{bonyh12}. Of course in the dissipative case we cannot diagonalize the non-selfadoint operator $\Ac$ as in \eqref{diag-Ac0}.

Even if we cannot use this nice reduction to a problem on $L^2$, we could use the norms \eqref{norm-HcPg-delta} for our problem. We have chosen the norms \eqref{norm-Hc-delta} instead. It is easy to see that this gives two equivalent norms when $\d = 0$. This is not so clear for $\d \neq 0$. Since we study the localisation of the energy for large times, we prefer the norm which involves the local operator $\nabla$ rather than the norm defined with the non-local operator $\sqP$. Thus, even if our purpose is to deal with the dissipative case, the estimates in $\Lc(\Hc^{\d},\Hc^{-\d})$ are interesting even in the self-adjoint setting.\\

Moreover, even if this does not play any role in Theorem \ref{th-loc-decay-improved}, we notice that we have improved the weight in the low frequency estimates for $N$ small. Again, this gives a more natural result than in \cite{boucletr14}. We recall that the weight $\pppg x \inv$ is sharp to obtain uniform estimates for the resolvent of the Schr\"odinger operator $(P-z)\inv$ (see \cite{boucletr15}). Thus we will have to use accurately the structure of the wave operator and of the energy space (based on the fact that we estimate the derivatives of the solution and not the $L^2$-norm of $u$ itself) to get a uniform estimate with weight $\pppg x ^{-\d}$ for any $\d > \frac 12$ (see Remark \ref{rem-weight}).\\

In our analysis we will have to use an improved version of the uniform estimates given by the Mourre commutators method. We recall that given a maximal dissipative operator $H$ on some Hilbert space $\Hc$, the idea of the Mourre method is to prove uniform estimates for
\[
\pppg A^{-\d} (H-z)\inv \pppg A^{-\d},
\]
where $\Im(z) > 0$, $\Re(z)$ belongs to some interval of $\R$, $A$ is a (self-adjoint) conjugate operator (in a sense given by Definition \ref{def-mourre}) and $\d > \frac 12$. The main assumption is a (spectrally localized) lower bound on the commutator between the self-adjoint part of $H$ and $A$. The original result (for a self-adjoint operator $H$) has been proved in \cite{mourre81}. The uniform estimates for the powers of the resolvent have been proved in \cite{jensenmp84,jensen85}: under additionnal assumptions on the multiple commutators between $H$ and $A$ we can show uniform estimates for 
\[
\pppg A^{-\d} (H-z)^{-1-N} \pppg A^{-\d}, \qquad \text{where} \quad  \d > N+\frac 12.
\]
Then there have been a lot of improvements in many directions. We refer to \cite{amrein} for a general overview on the subject. 

The case of a dissipative operator $H$ has been studied in \cite{art-mourre,boucletr14,art-mourre-formes}. We proved in particular that we can insert some operators between the resolvents. If the operators $\Phi_1,\dots,\Phi_N$ have good commutation properties with $A$, then we can generalize the estimate above for the operator
\begin{equation} \label{inserted-factors}
\pppg A^{-\d} (H-z)\inv \Phi_1 (H-z)\inv \Phi_2 (H-z)\inv \dots \Phi_N (H-z)\inv \pppg A^{-\d}.
\end{equation}
This was useful since the derivatives of the resolvent $R(z)$ are not given by its powers. For instance we have for the first derivative
\begin{equation} \label{der-R}
R'(z) = R(z) (ia + 2z) R(z).
\end{equation}
When $a = 0$, this is $2z R(z)^2$, but in the general case we have to use the Mourre method with the inserted operator $a(x)$.

Here we follow the same general idea to see that we can weaken the weights on both sides in \eqref{inserted-factors} if some of the inserted operators $\Phi_j$ are themselves of the form $\pppg A^{-\d_j}$. For instance we can prove uniform estimates for an operator of the form
\begin{equation} \label{ex-super-mourre}
\pppg A^{-(\d-1)} (H-z)^{-k} \pppg A \inv (H-z)^{-1-N +k}(z) \pppg A^{-(\d-1)}.
\end{equation}
Here $k \in \Ii 1 N$ and $\d$ is greater than $N+\frac 12$ as before. 

We know that for the resolvent of a Schr\"odinger operator we often use the generator of dilations as the conjugate operator, and then we replace the weights $\pppg A^{-\d}$ by $\pppg x^{-\d}$. Here we are going to use the decay of the absorption index $a$ in \eqref{der-R} to play the role of a weight. As a consequence, weaker weights will be required on the left and on the right, which will be crucial in our analysis.

We remark in \eqref{ex-super-mourre} that if we add one power of $\pppg A \inv$ between the resolvents we can remove one power of $\pppg A \inv$ on both sides. We will generalize this by adding more powers of $\pppg A \inv$ between the resolvents (see Theorem \ref{th-mourre}).\\

This paper is organized as follows. In Section \ref{sec-mourre} we state and prove this new version of the Mourre method with inserted weights in the abstract setting. Then in Section \ref{sec-res-estim-L2} we improve the results of \cite{boucletr14} concerning the estimates of the derivatives of $R(z)$ on $L^2$. Then we deduce in Section \ref{sec-res-estim-energy} the estimates of Theorems \ref{th-inter-freq}, \ref{th-high-freq} and \ref{th-low-freq} concerning the resolvent of $\Ac$ in the energy space $\Hc$. Finally we use these resolvent estimates in Section \ref{sec-loc-decay} to prove Theorem \ref{th-loc-decay-improved}.

\section{Mourre's method with inserted weights} \label{sec-mourre}

\newcommand{\Hco}{\Hc_0}

In this section we show how to insert weights between the resolvents in the estimates given by the Mourre method. We first recall the abstract setting. Even if we will only consider dissipative perturbations in the sense of operators as in \cite{art-mourre,boucletr14} we introduce the more general setting of perturbations in the sense of forms as described in \cite{art-mourre-formes}.\\

Let $\Hco$ be a Hilbert space. We recall that the operator $H$ on $\Hco$ with domain $\Dom(H)$ is said to be dissipative if 
\[
\forall \f \in \Dom(H), \quad \Im \innp{H\f}{\f} \leq 0.
\]
In this case we say that $H$ is maximal dissipative if it has no other dissipative extension than itself. This is equivalent to the fact that $(H-z)$ has a bounded inverse for some (and hence any) $z \in \C_+$.\\

Let $q_0$ be a quadratic form closed, densely defined, symmetric and bounded below on $\Hco$, with domain $\Kc = \Dom(q_0)$. Let $q_\Th$ be another symmetric form on $\Hco$, non-negative and $q_0$-bounded. Let $q = q_0 - i q_\Th$. Let $H_0$ be the self-adjoint operator corresponding to $q_0$ and $H$ be the maximal dissipative operator corresponding to $q$. We denote by $\tilde H : \Kc \to \Kc^*$ the operator which satisfies $q(\f,\p) = \innp{\tilde H \f} \p _{\Kc^*,\Kc}$ for all $\f,\p \in \Kc$. We similarly introduce the operators $\tilde H_0,\Th \in \Lc(\Kc,\Kc^*)$ corresponding to the forms $q_0$ and $q_\Th$, respectively. By the Lax-Milgram Theorem, the operator $(\tilde H -z)$ has a bounded inverse in $\Lc(\Kc^*,\Kc)$ for all $z \in \C_+$. Moreover for $\f \in \Hco \subset \Kc^*$ we have $(H-z)\inv \f = (\tilde H - z)\inv \f$.\\

Now we introduce the conjugate operator for $H$. To simplify the discussion, we consider an operator which is conjugate to $H$ at any order:

\begin{definition} \label{def-mourre}
Let $A$ be a self-adjoint operator on $\Hco$. We say that $A$ is a conjugate operator (in the sense of forms) to $H$ on the interval $J$, at any order, and with bounds $\a \in ]0,1]$, $\b \geq 0$ and $\Upsilon_N \geq 0$ for $N \in \N^*$ if the following conditions are satisfied:
\begin{enumerate} [(i)]
\item \label{item-Kc-invariant}
The form domain $\Kc$ is left invariant by $e^{-itA}$ for all $t \in \R$. We denote by $\Ec$ the domain of the generator of $\restr{e^{-itA}}{\Kc}$.
\item The commutators $B^0 = [\tilde H_0,iA]$ and $B_1 = [\tilde H,iA]$, \emph{a priori} defined as operators in $\Lc(\Ec,\Ec^*)$, extend to operators in $\Lc(\Kc,\Kc^*)$. Then for all $n \in \N^*$ the operator $[B_{n},iA]$ defined (inductively) in $\Lc(\Ec,\Ec^*)$ extends to an operator in $\Lc(\Kc,\Kc^*)$, which we denote by $B_{n+1}$. 
\item For all $N \in \N^*$ we have 
\[
\nr {B_1} \leq \sqrt \a \Upsilon_N \qandq \nr {B_1 + \b \Th} \nr{B^0}  + \b \nr{[\Th,A]} + \sum_{n=2}^{N+1} {\nr{B_n}}  \leq \a \Upsilon_N ,
\]
where all the norms are in $\Lc(\Kc,\Kc^*)$.
\item We have 
\begin{equation} \label{hyp-mourre}
\1 J (H_0) (B^0 + \b \Th) \1 J (H_0) \geq \a \1 J (H_0).
\end{equation}
\end{enumerate}
\end{definition}

Let $\n \in \N$ and $n_0,\dots,n_\n \in \N^*$. Let $j \in \Ii 0 \n$. We consider $\Phi_{j,0} \in \Lc(\Kc,\Hco)$, $\Phi_{j,1},\dots,\Phi_{j,n_j-1} \in \Lc(\Kc,\Kc^*)$ and $\Phi_{j,n_j} \in \Lc(\Hco,\Kc^*)$. We assume (inductively) on $m \in \N$ that the operator 
\[
\ad_{iA}^m (\Phi_{j,0}) := [\ad_{iA}^{m-1} (\Phi_{j,0}),iA]
\]
(with $\ad_{iA}^0 (\Phi_{j,0}) = \Phi_{j,0}$), at least defined as an operator in $\Lc(\Ec,\Ec^*)$, can be extended to an operator in $\Lc(\Kc,\Hco)$. We assume similarly that the commutators $\ad_{iA}^m (\Phi_{j,k})$ for $m \in \N$ and $k \in \Ii 1 {n_j-1}$ extend to operators in $\Lc(\Kc,\Kc^*)$, and finally that the commutators $\ad_{iA}^m (\Phi_{j,n_j})$ for $m \in \N$ extend to operators in $\Lc(\Hco,\Kc^*)$. Then for $k \in \Ii 1 {n_j-1}$ and $N \in \N^*$ we set
\[
\nr{\Phi_{j,k}}_{\Cc_N(A,\Kc,\Kc^*)} = \sum_{m=0}^N \nr{\ad_{iA}^m (\Phi_{j,k})}_{\Lc(\Kc,\Kc^*)}.
\]
We similarly define $\nr{\Phi_{j,0}}_{\Cc_N(A,\Kc,\Hco)}$ and $\nr{\Phi_{j,n_j}}_{\Cc_N(A,\Hco,\Kc^*)}$, and then
\[
\nr{(\Phi_{j,0},\dots,\Phi_{j,n_j})}_{\Cc_N^{n_j}(A)} = \nr{\Phi_{j,0}}_{\Cc_N(A,\Kc,\Hco)} \nr{\Phi_{j,n_j}}_{\Cc_N(A,\Hco,\Kc^*)}\prod_{k=1}^{n_j-1} \nr{\Phi_{j,k}}_{\Cc_N(A,\Kc,\Kc^*)}.
\]
For $z \in \C_+$ we set  
\begin{equation} \label{def-Rcj}
\Rc_j(z) = \Phi_{j,0} (\tilde H-z)\inv  \Phi_{j,1} (\tilde H-z)\inv \dots \Phi_{j,n_j-1} (\tilde H-z)\inv \Phi_{j,n_j} \quad \in \Lc(\Hco).
\end{equation}

In \cite{art-mourre-formes} we have proved uniform estimates for $\Rc_j(z)$. In particular for $I \Subset J$ and $\d > n_j - \frac 12$ there exists $C \geq 0$ (which only depends on the relative bound of $q_\Th$ with respect to $q_0$, $\d$, $J$, $I$ and the constants $\a$, $\b$ and $\Upsilon_N$ which appear in Definition \ref{def-mourre}) such that for
\[
z \in \C_{I,+} := \singl{z \in \C_+ \st \Re(z) \in I}
\]
we have 
\[
\nr{\pppg {A}^{-\d} \Rc_j(z) \pppg A ^{-\d}}_{\Lc(\Hco)} \leq C \nr{(\Phi_{j,0},\dots,\Phi_{j,\n_j})}_{\Cc_{\n_j}^{\n_j}(A)}.
\]
Let $\d_1,\dots,\d_{\n} \in \R$ and for $z \in \C_+$:
\begin{equation} \label{def-Rc}
\Rc(z) = \Rc_0(z) \pppg A^{-\d_1} \Rc_1(z) \dots \pppg A^{-\d_{\n}} \Rc_\n(z)  \quad \in \Lc(\Hco) .
\end{equation}
If $\d_1,\dots,\d_{\n}$ are non-negative we can consider $\Phi_{j-1,n_{j-1}} \pppg A^{-\d_j} \Phi_{j,0} \in \Lc(\Kc,\Kc^*)$ for $j \in \Ii 1 {\n}$ as an inserted factor, and we directly obtain a uniform bound for 
\[
\pppg A^{-\d} \Rc(z) \pppg A^{-\d}, \quad \text{where } \d > n_0 + \dots + n_\n - \frac 12.
\]
Our porpose is to use the inserted weights $\pppg A^{-\d_j}$ to weaken the weights $\pppg A^{-\d}$ on both sides. For this we will use the following lemma:

\begin{lemma} \label{lem-super-mourre} 
Let $\nu \in \N$. Let $\Rc_0,\dots,\Rc_\nu \in \Lc(\Hco)$. Let $Q$ be a self-adjoint operator on $\Hco$ with $Q \geq 1$. Let $P_-,P_+ \in \Lc(\Hco)$ be such that $P_+ + P_- = \Id_{\Hco}$. Assume that $Q$ commutes with $P_-$ and $P_+$. Let $n_0,\dots,n_\nu \in \N^*$ and $\g_0,\dots,\g_\nu \in \R_+$. Assume that for $j \in \Ii 0 \nu$, $\s > n_j - \frac 12$, $\s_l \geq 0$ and $\s_r \geq 0$ there exists $c \geq 0$ such that
\begin{enumerate}[(i)]
\item $\nr{Q^{-\s} \Rc_j Q^{-\s}} \leq c \g_j$,
\item $\nr{Q^{\s-n_j} P_- \Rc_j Q^{-\s}} \leq c \g_j$,
\item $\nr{Q^{-\s} \Rc_j P_+ Q^{\s - n_j}} \leq c \g_j$,
\item $\nr{Q^{\s_l} P_- \Rc_j P_+ Q^{\s_r}} \leq c \g_j$.
\end{enumerate}
Let $\d_0,\dots,\d_{\nu+1},\d_-,\d_+ \in \R$ and $\d_l,\d_r \in \R_+$ be such that 
\begin{equation} \label{hyp-d-+}
\forall k \in \Ii 0 {\nu}, \quad \sum_{j=0}^k \d_j > \sum_{j=0}^{k} n_j - \frac 12 \quad \text{and} \quad  \sum_{j=k+1}^{\nu+1} \d_j > \sum_{j=k}^{\nu} n_j - \frac 12,
\end{equation}
and moreover:
\[
\d_+ \geq \sum_{j=0}^{\nu} n_j  - \sum_{j=0}^{\nu} \d_j\quad \text{and} \quad \d_-   \geq \sum_{j=0}^{\nu} n_j -  \sum_{j=1}^{\nu+1} \d_j.
\]
Let 
\[
\Rc = \Rc_0 Q^{-\d_1} \Rc_1 Q^{-\d_2} \dots \Rc_{\nu-1} Q^{-\d_{\nu}} \Rc_\nu .
\]
Then there exists $C \geq 0$ which only depends on $\nu$, $\d_0,\dots,\d_{\nu+1},\d_-,\d_+,\d_l,\d_r$ and the constants $c$ which appear in (i)-(iv) such that 
\begin{enumerate}[(I)]
\item $\nr{Q^{-\d_0} \Rc Q^{-\d_{\nu+1}} } \leq C \g_0\dots\g_\nu$.
\item $\nr{Q^{- \d_-} P_- \Rc Q^{-\d_{\nu+1}}} \leq C \g_0\dots\g_\nu$,
\item $\nr{Q^{-\d_0} \Rc P_+ Q^{-\d_+}} \leq C \g_0\dots\g_\nu$,
\item $\nr{Q^{\d_l} P_- \Rc P_+ Q^{\d_r}} \leq C \g_0\dots\g_\nu$.
\end{enumerate}
\end{lemma}

Notice that the assumptions on the exponents $\d_0,\dots,\d_{\nu+1},\d_-,\d_+$ do not imply that they are all non negative. This abstract lemma will be applied in the context of the Mourre method as follows: $Q$ is the weight $\pppg A$, $P_+$ and $P_-$ are the spectral projections $\1 {\R_+}(A)$ and $\1 {\R_-}(A)$, respectively, and $\Rc_{j}$ is as in \eqref{def-Rcj}: the product of $n_j$ resolvents, maybe with inserted factors. The assumptions (i)-(iv) will be consequences of the dissipative Mourre theory (see Theorem 5.16 in \cite{boucletr14}), and we will use estimate (I) with inserted weights.

\begin{proof}[Proof of Lemma \ref{lem-super-mourre}]
The result is given by the assumptions if $\nu = 0$. We prove the general case by induction on $\nu$. Let 
\[
\Rc^* = \Rc_0 Q^{-\d_1} \Rc_1 Q^{-\d_2} \dots \Rc_{\nu-2} Q^{-\d_{\nu-1}} \Rc_{\nu-1},
\]
so that $\Rc= \Rc^* Q^{-\d_{\nu}} \Rc_\nu$. We set 
\[
\d_* = \sum_{j=0}^{\nu-1} \d_j \qandq n_* = \sum_{j=0}^{\nu-1} n_j.
\]
We have 
\[
Q^{-\d_0} \Rc Q^{-\d_{\nu+1}} = Q^{-\d_0} \Rc^* P_+  Q^{-\d_{\nu}} \Rc_\nu Q^{-\d_{\nu+1}} + Q^{-\d_0} \Rc^*  Q^{-\d_{\nu}} P_- \Rc_\nu Q^{-\d_{\nu+1}}
\]
and hence 
\begin{align*}
\nr{Q^{-\d_0} \Rc Q^{-\d_{\nu+1}}}
& \leq \nr{Q^{-\d_0} \Rc^* P_+ Q^{\d_* - n_*}}   \nr{ Q^{n_* - \d_* - \d_{\nu}} \Rc_\nu Q^{-\d_{\nu+1}}} \\
& \quad + \nr{Q^{-\d_0} \Rc^*  Q^{n_\nu - \d_{\nu+1} -\d_{\nu}} } \nr{Q^{\d_{\nu+1} - n_\nu} P_- \Rc_\nu Q^{-\d_{\nu+1}}} \\
& \lesssim \g_0 \dots \g_{\nu-1} \times \g_\nu.
\end{align*}
We have used (III)$_{\nu-1}$ with $\d_+ = n_* - \d_*$, (i) with $\s = \min(\d_* + \d_{\nu} - n_* , \d_{\nu+1}) > n_\nu - \frac 12$, (I)$_{\nu-1}$ with $\d_{\nu}$ replaced by $\d_{\nu} + \d_{\nu+1} - n_\nu$ and finally (ii) with $\s = \d_{\nu+1}$. This proves (I)$_\nu$. We follow the same idea for the other estimates. We have 
\begin{align*}
\nr{Q^{- \d_-} P_- \Rc Q^{-\d_{\nu+1}}}
& \leq \nr{Q^{-\d_-} P_- \Rc^* Q^{n_\nu - \d_{\nu+1} - \d_{\nu}}} \nr{Q^{\d_{\nu+1}-n_\nu} P_- \Rc_\nu Q^{-\d_{\nu+1}}} \\
& \quad + \nr{Q^{-\d_-} P_- \Rc^* P_+ Q^{n_\nu  - \d_{\nu}}} \nr{ Q^{-n_\nu} \Rc_\nu Q^{-\d_{\nu+1}}}\\
& \lesssim  \g_0 \dots \g_{\nu-1} \times \g_\nu.
\end{align*}
For the third estimate we write:
\begin{align*}
\nr{Q^{- \d_0}  \Rc P_+ Q^{-\d_+}}
& \leq \nr{Q^{-\d_0} \Rc^* P_+ Q^{\d_* - n_*} } \nr{Q^{n_* - \d_* - \d_{\nu}} \Rc_\nu P_+ Q^{-\d_+}} \\
& \quad + \nr{Q^{-\d_0} \Rc^* Q^{n_\nu - \d_{\nu+1} - \d_{\nu}}} \nr{ Q^{\d_{\nu+1}-n_\nu} P_- \Rc_\nu P_+ Q^{-\d_{\nu+1}}}\\
& \lesssim  \g_0 \dots \g_{\nu-1} \times \g_\nu.
\end{align*}
And finally for $\d$ large enough:
\begin{align*}
\nr{Q^{\d_l} P_- \Rc P_+ Q^{\d_r}}
& \leq \nr{Q^{\d_l} P_- \Rc^* Q^{- \d}} \nr{ Q^{\d} P_- \Rc_\nu P_+ Q^{\d_r}}\\
& \quad +  \nr{Q^{\d_l} P_- \Rc^* P_+ Q^{\d}} \nr{ Q^{-\d} \Rc_\nu P_+ Q^{\d_r}}\\
& \lesssim  \g_0 \dots \g_{\nu-1} \times \g_\nu.
\end{align*}
This concludes the proof of the lemma.
\end{proof}

Starting from the results of \cite{art-mourre-formes} we can then deduce the following improved version of the dissipative Mourre method:

\begin{theorem} \label{th-mourre}
Let $H$ be a maximal dissipative operator on $\Hco$ as described at the beginning of the section. Let $A$ be a conjugate operator in the sense of Definition \ref{def-mourre}. Let $\Phi_{j,k}$ for $j \in \Ii 0 \n$ and $k \in \Ii 0 {n_j}$ be as above. Let $I$ be a compact subset of $\mathring J$. Let $\d_0,\dots,\d_{\nu+1},\d_-,\d_+ \in \R$ and $\d_l,\d_r \in \R_+$ be as in Lemma \ref{lem-super-mourre}. Let $\Rc(z)$ be as in \eqref{def-Rc}. Then there exists $C \geq 0$ such that for all $z \in \C_{I,+}$ we have 
\[
\nr{\pppg A^{-\d_0} \Rc(z) \pppg A^{-\d_{\n+1}}} _{\Lc(\Hco)} \leq C \prod_{j=0}^\n \nr{(\Phi_{j,0},\dots,\Phi_{j,n_j})}_{\Cc_{n_j}^{n_j}}, 
\]
\[
\nr{\pppg A^{-\d_-} \1 {\R_-}(A) \Rc(z) \pppg A^{-\d_{\n+1}}} _{\Lc(\Hco)} \leq C \prod_{j=0}^\n \nr{(\Phi_{j,0},\dots,\Phi_{j,n_j})}_{\Cc_{n_j}^{n_j}}, 
\]
\[
\nr{\pppg A^{-\d_0} \Rc(z) \1{\R_+}(A) \pppg A^{-\d_+}} _{\Lc(\Hco)} \leq C \prod_{j=0}^\n \nr{(\Phi_{j,0},\dots,\Phi_{j,n_j})}_{\Cc_{n_j}^{n_j}}, 
\]
and
\[
\nr{\pppg A^{\d_l} \1 {\R_-}(A)\Rc(z)\1 {\R_+}(A) \pppg A^{\d_r}} _{\Lc(\Hco)} \leq C \prod_{j=0}^\n \nr{(\Phi_{j,0},\dots,\Phi_{j,n_j})}_{\Cc_{n_j}^{n_j}}.
\]
\end{theorem}

We will use this result in the proof of Proposition \ref{prop-estim-res-amort}. We recall that when $H$ is (a perturbation of) the free Laplacian on $L^2(\R^d)$ we usually consider (a perturbation of) the generator of dilations as the conjugate operator:
\[
 A = - \frac i2 ( x \cdot \nabla + \nabla \cdot x) = - i \, (x \cdot \nabla) - \frac {id}2 .
\]
We record in the following proposition the properties of $A$ we are going to use in this paper:
\begin{proposition} \label{prop-A}
\begin{enumerate}[(i)]
\item For $\th \in \R$, $u \in \Sc$ and $x \in \R^d$ we have
\begin{equation*} %\label{expA}
 (e^{i\th A} u) (x) = e^{\frac {d\th}2} u ( e^\th x).
\end{equation*}
\item For $j \in \Ii 1 d$ and $\g \in C^\infty(\R^d)$ we have on $\Sc$:
\begin{equation*} %\label{comm-A}
 [\partial_j , i A] = \partial_j \quad \text{and} \quad [\g,iA] = -(x \cdot \nabla) \g.
\end{equation*}
 \item For $p \in [1,+\infty]$, $\th \in \R$ and $u \in \Sc$ we have
\[
 \nr{e^{i\th A} u}_{L^p} = e^{\th \left( \frac d 2 - \frac d p\right)} \nr u_{L^p}.
\]
\end{enumerate}
\end{proposition}

\section{Uniform resolvent estimates in $L^2$.} \label{sec-res-estim-L2}

In this section we prove on $L^2$ the uniform estimates for the derivatives of the resolvent $R(z)$ defined by \eqref{def-R}. This will be used in the next section to obtain estimates on $(\Ac - z)^{-N-1} \in \Lc(\Hc)$.
We first recall from \cite[Proposition 5.9]{boucletr14} that the derivatives of $R(z)$ are not the powers thereof:

\begin{proposition} \label{prop-der-R2}
Let $N \in \N$. Then the derivative $R^{(N)}(z)$ is a linear combination of terms of the form 
\begin{equation} \label{terme-dec-RN}
z^{\o} R(z) a(x)^{\nu_1} R(z)  a(x)^{\nu_2} \dots  a(x)^{\nu_n} R(z),
\end{equation}
where $n \in \Ii 0 N$ (there are $n+1$ factors $R(z)$), $\o \in \N$ and $\nu_1,\dots,\nu_n \in \{ 0,1\}$ are such that 
\begin{equation} \label{eq-dec-RN}
N = 2n - \o - \Vc, \quad \text{where } \Vc = \nu_1 + \dots + \nu_n.
\end{equation}
\end{proposition}

\subsection{Intermediate and high frequency estimates}

For intermediate and high frequencies we will use directly the estimates given in \cite{boucletr14}. We only have to check that we can add derivatives on both sides of the resolvent:

\begin{proposition} \label{prop-high-freq}
Let $N \in \N$ and $\d > N + \frac 12$. Let $\nul, \nur \in \N^d$ be such that $\abs \nul \leq 1$ and $\abs \nur \leq 1$. Let $\g \in ]0,1]$.
\begin{enumerate}[(i)]
\item There exists $C \geq 0$ such that for all $z \in \C_+$ with $\g \leq \abs z \leq \g\inv$ we have 
\[
\nr{\pppg x ^{-\d} D^{\nul} R^{(N)}(z) D^{\nur} \pppg x^{-\d}}_{\Lc(L^2)} \leq C.
\]
\item Assume that the damping condition \eqref{hyp-damping} is satisfied. Then there exists $C \geq 0$ such that for all $z \in \C_+$ with $\abs z \geq \g$ we have 
\[
\nr{\pppg x ^{-\d} D^{\nul} R^{(N)}(z) D^{\nur} \pppg x^{-\d}}_{\Lc(L^2)} \leq C \abs z^{\abs \nul + \abs \nur - 1}.
\]
\end{enumerate}
\end{proposition}

\begin{proof}
We assume that \eqref{hyp-damping} holds and prove the second statement. The proof of the first statement is analogous. Let $\h \in C_0^\infty(\R,[0,1])$ be equal to 1 on $[-2,2]$. For $z \in \C_+$ we set $\h_z : \t \mapsto \h \big( \t / {\abs z^2} \big)$. By pseudo-differential and functional calculus, the operators
\[
D^\nul (1+ \Pg)^{-\frac {\abs \nul}2}, \quad 
\left(\frac {1+ \Pg}{\abs z^2}\right)^{\frac {\abs \nul}2} \left(1+ \frac {\Pg}{\abs z^2}\right)^{-\frac {\abs \nul}2} \qandq  \left(1+ \frac {\Pg}{\abs z^2}\right)^{\frac {\abs \nul} 2}\h_z(\Pg)
\]
are uniformly bounded in $L^{2,-\d}$ for $\abs z \geq \g$, so
\[
\nr{\pppg x^{-\d} D^\nul \h_z(\Pg) \pppg x^\d} \lesssim \abs z^{\abs \nul}.
\]
With a similar estimate on the right and according to Theorem 1.5 in \cite{boucletr14} we obtain 
\begin{eqnarray*}
\lefteqn{\nr{\pppg x ^{-\d} D^{\nul} \h_z(\Pg) R^{(N)}(z) \h_z(\Pg) D^{\nur} \pppg x^{-\d}}}\\
&&\leq \nr{\pppg x^{-\d} D^\nul \h_z(\Pg) \pppg x^\d} \nr{\pppg x ^{-\d} R^{(N)}(z) \pppg x^{-\d}} \nr{\pppg x^{\d}  \h_z(\Pg) D^\nur \pppg x^{-\d}}\\
&& \lesssim \abs z^{\abs \nul + \abs \nur -1}.
\end{eqnarray*}
For $z \in \C_+$ we set $R_0(z) = (\Pg-z^2)\inv$. As above we obtain
\[
\nr{\pppg x^{-\d} D^\nul (1-\h_z)(\Pg) R_0(z) \pppg x^\d}_{\Lc(L^2)} \lesssim \abs z^{\abs \nul - 2}.
\]
We use the decomposition and the notation of Proposition \ref{prop-der-R2}. Let $T(z)$ be a term like \eqref{terme-dec-RN}. We set 
\[
\tilde T(z) = z^{\o} a(x)^{\nu_1} R(z)  a(x)^{\nu_2} \dots  a(x)^{\nu_n} R(z)
\]
($\tilde T(z) = \Id_{L^2}$ if $n = 0$). By the resolvent identity we have 
\begin{eqnarray*}
\lefteqn{\nr{\pppg x^{-\d} D^\nul (1-\h_z)(\Pg) T(z) \h_z(\Pg) D^\nur \pppg x^{-\d}}}\\
&& \leq \nr{\pppg x^{-\d} D^\nul (1-\h_z)(\Pg)R_0(z) \pppg x^\d} \nr{ \pppg x^{-\d} \big( \tilde T(z) +iz a^{\n_1}(x) T(z) \big) \pppg x^{-\d} }\\
&& \qquad \times \nr{\pppg x^\d \h_z(\Pg) D^\nur \pppg x^{-\d}}\\
&& \lesssim \abs z ^{\abs \nul + \abs \nur - 2}.
\end{eqnarray*}
We have used the fact that $\pppg x^{-\d} \big( \tilde T +iz a^{\n_1}(x) T \big) \pppg x^{-\d}$ is uniformly bounded, which can be proved exactly as Theorem 1.5 in \cite{boucletr14}. Similarly we prove that 
\[
\nr{\pppg x^{-\d} D^\nul (1-\h_z)(\Pg) T \h_z(\Pg) D^\nur \pppg x^{-\d}} \lesssim \abs z ^{\abs \nul + \abs \nur - 2}
\]
and 
\[
\nr{\pppg x^{-\d} D^\nul (1-\h_z)(\Pg) T (1-\h_z)(\Pg) D^\nur \pppg x^{-\d}} \lesssim \abs z ^{\abs \nul + \abs \nur - 2},
\]
which concludes the proof.
\end{proof}

\subsection{Low frequency estimates: statement of the results}

We now turn to the low frequency estimates. The following result improves and completes Theorem 1.3 in \cite{boucletr14}:

\begin{theorem} \label{th-low-freq-R}

\begin{enumerate}[(i)]
\item Let $N\in\N$. Let $\nul,\nur \in \N^d$ be such that $\anul \leq 1$ and $\anur \leq 1$. Assume that $\nur = 0$ or $b_1 = \dots = b_d = 0$. Let $\d > N + \frac 12$. Let $\e > 0$. If $\anul + \anur \geq 1$ or $N\neq 0$ then there exists a neighborhood $\Uc$ of 0 in $\C$ and $C \geq 0$ such that for all $z \in \Uc \cap \C_+$ we have 
\[
\nr{\pppg x^{-\d} D^{\nul} R^{(N)}(z) D^{\nur} \pppg x^{-\d}} _{\Lc(L^2)} \leq C \left( 1 + \abs z^{d- N - 2 + \abs{\nul} + \abs{\nur} - \e} \right).
\]
\item Let $\d > \frac 1 2$ and $\s \in \big[0,\frac 12 \big]$. Then there exists a neighborhood $\Uc$ of 0 in $\C$ and $C \geq 0$ such that for all $z \in \Uc \cap \C_+$ we have 
\[
\nr{\pppg x^{-\d-\s}  R(z)  \pppg x^{-\d-\s}} _{\Lc(L^2)} \leq \frac C {\abs z^{1-2\s}}.
\]
\end{enumerate}
\end{theorem}

\begin{remark} \label{rem-weight}
The second statement will only be used with $\s = 0$, in which case we have an estimate of size $\abs z\inv$. With $\s = \frac 12$ we have $\d > 1$ and a uniform estimate, so Theorem \ref{th-low-freq-R} contains Theorem 1.3 in \cite{boucletr14}. We see that for a single resolvent the weight $\pppg x^{-\d}$ with $\d > 1$ remains almost optimal (as mentioned in introduction, the optimal weight is in fact $\pppg x \inv$, see \cite{boucletr15}). However in our analysis for the wave equation the resolvent $R(z)$ will never come alone. It will either be composed with a derivative (in which case we can use the first statement of the theorem), or multiplied by $\abs z$ (in this case an estimate of size $\abs z \inv$ is enough). This explains why for $N = 0$ it is enough to assume that $\d > \frac 12$ in Theorem \ref{th-low-freq}.
\end{remark}

\begin{remark}
Compared to Theorem 1.3 in \cite{boucletr14} we allow a derivative on the right. This is easier if $W = 0$, which explains the assumption $\nur = 0$ or $b_1 = \dots = b_d = 0$. In fact, Theorem \ref{th-low-freq-R} will only be used in the proof of Proposition \ref{prop-res-Agg}, for which we have $W = 0$. We take the contribution of $W$ into account in Proposition \ref{prop-Agg-Ac}.
\end{remark}

\begin{remark} \label{estim-a-nul}
For these low frequency estimates there is no damping assumption, so the same estimates hold when $a = 0$.
\end{remark}

Theorem \ref{th-low-freq-R} will be a consequence of Proposition \ref{prop-der-R2} and the following result:

\begin{proposition} \label{prop-low-freq-R}
Let $n\in\N$. Let $\nul,\nur \in \N^d$ be such that $\anul \leq 1$ and $\anur \leq 1$. Assume that $\nur = 0$ or $b_1 = \dots = b_d = 0$. 
Let $\nu_1,\dots,\nu_n \in \{0,1\}$ and $\Vc = \nu_1+\dots+\nu_n$.
For $z \in \C_+$ we set
\[
\Rc(z) = R(z) a(x)^{\n_1}R(z) a(x)^{\n_2}  \dots R(z) a(x)^{\n_n} R(z).
\]
Let
\begin{equation} \label{hyp-delta}
\d_l > n + \frac 12 - \Vc \qandq \d_r > n + \frac 12 - \Vc.
\end{equation}
Set 
\[
s_l = \min\left( \frac {d}2 , \d_l \right) \qandq s_r = \min\left( \frac {d}2 , \d_r \right),
\]
and let $S < s_l + s_r$. Then there exists a neighborhood $\Uc$ of $0$ in $\C$ and $C \geq 0$ such that for all $z \in \Uc \cap \C_+$ we have
\[
\nr{\pppg x^{-\d_l} D^{\nul} \Rc(z) D^{\nur} \pppg x^{-\d_r}} _{\Lc(L^2)} \leq C  \left( 1 + \abs z^{- 2(n+1) + S  + \Vc + \anul + \anur}\right).
\]
\end{proposition}

\begin{proof}[Proof of Theorem \ref{th-low-freq-R}]
Let $T(z)$ be a term like \eqref{terme-dec-RN}. We can apply Proposition \ref{prop-low-freq-R} with $S = \min(d-\e,2N+1)$. With \eqref{eq-dec-RN} we obtain
\begin{align*}
\nr{T(z)}
& \lesssim 1+\abs z^{\o  - 2(n+1)+ d -\e  + \Vc + \anul + \anur } + \abs{z}^{\o + 2(N-n) - 1 + \Vc + \anul + \anur}\\
& \lesssim 1+\abs z^{d - N - 2  + \anul + \anur - \e} + \abs{z}^{N - 1 + \anul + \anur}.
\end{align*}
If $\abs \nul + \abs \nur \geq 1$ or $N \geq 1$ we obtain 
\[
\nr{T(z)} \lesssim 1 + \abs z^{d - N - 2  + \anul + \anur - \e} .
\]
This proves the first statement of the theorem.
Now assume that $N=0$ (and hence $n=0$). Let $\s \in \big[0,\frac 12\big]$. We apply proposition \ref{prop-low-freq-R} with $S = 1 +2\s$, and we get the second statement. This concludes the proof of Theorem \ref{th-low-freq-R}.
\end{proof}

In the proof of Theorem \ref{th-low-freq-R} we have only used Proposition \ref{prop-low-freq-R} with $\d_l, \d_r > n + \frac 12$. Now we use the refined assumption \eqref{hyp-delta}. The following result will be used in the proof of Proposition \ref{prop-res-Agg}:

\begin{corollary} \label{cor-sigma-r-un}
Let $\e > 0$. Let $T(z)$ be a term given by Proposition \ref{prop-der-R2}. Assume that $n = N$ and $\Vc \geq 1$. Let $\d > N + \frac 12$ and $j \in \Ii 1 d$. Then there exists $C \geq 0$ such that for $z \in \C_+$ small enough we have 
\[
\nr {\pppg x^{-(\d-1)} D_j T(z) \pppg x^{-(\d-1)}}_{\Lc(L^2)} \leq C \left( 1+ \abs z^{N-2} + \abs z ^{d-N-1-\e} \right).
\]
\end{corollary}

\begin{proof}
We apply Proposition \ref{prop-low-freq-R} with $S = \min(d-\e, 2n-1)$. 
If $S = {d-\e}$ we have
\[
\nr {\pppg x^{-(\d-1)} D_j T(z) \pppg x^{-(\d-1)}}_{\Lc(L^2)} \lesssim 1 + \abs{z}^{\o + d-\e - 2n - 1 + \Vc} \lesssim 1 + \abs{z}^{d-\e-N-1}.
\]
And if $S = 2n-1$:
\[
\nr {\pppg x^{-(\d-1)} D_j T(z) \pppg x^{-(\d-1)}}_{\Lc(L^2)} \lesssim  1 + \abs{z}^{\o - 2 + \Vc} \lesssim 1 + \abs{z}^{N-2}. \qedhere
\]
\end{proof}

It the rest of this section we prove Proposition \ref{prop-low-freq-R}. We first use a scaling argument to prove the result for the resolvent of a small perturbation of the free Laplacian (see Proposition \ref{prop-low-freq-iota}) and then we prove the general case by a perturbation argument.\\

Let $\h \in C_0^\infty(\R^d,[0,1])$ be equal to 1 on a neighborhood of 0. For $\y \in ]0,1]$ and $x \in \R^d$ we set $\h_\y(x) = \h(\y x)$ and $G_\y(x) = \h_\y(x) I_d + (1-\h_\y(x)) G(x)$. Then we consider the operators
\begin{equation} \label{def-Pii-Pcc}
\Pii = -\divg G_\y(x) \nabla \quad \text{and} \quad  \Pcc = \Pg - W - \Pii = -\divg \big(\h_\y(x) (G(x)-I_d)\big) \nabla .
\end{equation}
For the dissipative part we set $a_\y = (1 - \h_\y) a$, and finally, for $z \in \C_+$:
\[
\Hii = \Pii - i z a_\y
\quad 
\text{
and 
}
\quad
\Ri (z) = (\Hii - z^2)\inv.
\]

\subsection{Low frequency resolvent estimates for a small perturbation of the free case}

\newcommand{\sss}{\s}

We first prove Proposition \ref{prop-low-freq-R} with $R(z)$ replaced by $\Ri(z)$ (see Proposition \ref{prop-low-freq-iota}).  For this we use a scaling argument. For a function $u$ on $\R^d$ and $z \in \C^*$ we denote by $u_z$ the function 
\[
u_z : x \mapsto u \left( \frac x {{\abs z}} \right).
\]
The analysis is based on the fact that the multiplication by a decaying function somehow behaves like a derivative, and hence is small in the low frequency regime (this is in the spirit of the Hardy inequality). More precisely, for $\d \in \R$ we define $\symb^{-\d}$ as the set of smooth functions $\vf$ such that 
\[
\abs{\partial^\a \vf(x)} \leq c_\a \pppg x^{-\d - \abs \a}.
\]
For $\n \geq 0$, $N \in \N$ and $\vf \in \Sc^{-\n - \frac \rho 2}(\R^d)$ we set
\[
\nr{\vf}_{\n,N} = \sup_{\abs \a \leq d} \sum_{0\leq m \leq N} \sup _{x\in \R^d} \abs{\pppg x^{\n + \frac \rho 2 + \abs \a} \big( \partial^\a (x\cdot \nabla)^m \vf \big) (x)}.
\]
Then we have the following result:

\begin{proposition} \label{prop-dec-sob}
Let $\n \in \big[ 0, \frac d 2 \big[$ and $s \in \big] -\frac d 2, \frac d 2\big[$ be such that $s -\n \in \big] -\frac d 2, \frac d 2\big[$. Then there exists $C \geq 0$ such that for $\vf \in \symb^{-\n-\frac \rho 2}(\R^d)$, $u \in H^s$ and $\l > 0$ we have
\[
 \nr{\vf_\l u}_{\dot H^{s-\n}} \leq C \l ^\n \nr \vf_{\n,0} \nr u _{\dot H^s}
\]
and
\[
 \nr{\vf_\l u}_{H^{s-\n}} \leq C \l ^\n \nr \vf_{\n,0} \nr u _{H^s}.
\]
\end{proposition}

For the proof we refer to Proposition 7.2 in \cite{boucletr14}. Now for $z \in \C_{+}$ we set
\begin{align*} 
\tPii =  \frac 1 {\abs z^2} e^{-iA \ln \abs z} \Pii e^{iA \ln \abs z} = -\divg G_{\y, z}(x) \nabla 
\end{align*}
and 
\begin{equation*}
\tHii =  \frac 1 {\abs z^2} e^{-iA \ln \abs z} \Hii e^{iA \ln \abs z} = \tPii - i \frac {\hat z} { \abs z} a_{\y, z}
\end{equation*}
(where $\hat z$ stands for $z / \abs z$). Then we set
\begin{equation*} 
\Riz(z) = \big(\tHii- \hat z^2\big) \inv,
\end{equation*}
so that
\begin{equation*} 
\Ri(z)  = \frac 1 { \abs z^2} e^{i A \ln \abs z} \Riz(z) e^{-i A \ln  \abs z}.
\end{equation*}

For $\m = (\m_1,\dots,\m_d) \in \N^d$ we set 
\[
\ad_x^\m := \ad_{x_{1}}^{\m_1} \dots \ad_{x_{d}}^{\m_d}.
\]
Let $s \in \big]- \frac d 2 ,\frac d 2 - 1 \big[$, $m \in \N$ and $\m \in \N^d$. According to Propositions \ref{prop-A} and \ref{prop-dec-sob} there exists $C \geq 0$ such that for all $j \in \Ii 1 n$, $z \in \C_+$ and $u \in \Sc$ we have 
\[
\nr{ \ad_{iA}^m a_{\y,z} u}_{H^{s}} = \nr{\big((-x \cdot \nabla )^m a_{\y}\big)_z u}_{H^{s}} \leq C \abs z \nr{u}_{H^{s+1}}.
\]
If $\m \neq 0$ we have $\ad_x^\m \ad_{iA}^m a_{\y,z} u = 0$. For $z \in \C_+$ we set 
\[
\tPhil = e^{-iA \ln\abs z} D^{\nul} e^{iA \ln \abs z} = \abs z^{\abs \nul} D^\nul \quad \text{and}\quad \tPhir  = \abs z^{\abs \nur} D^\nur.
\]
The only property which we are going to use on these two operators is that for $s \in \big[0 ,\frac d 2\big[$, $m \in \N$ and $\m \in \N^d$ there exists $C \geq 0$ such that for $z \in \C_+$ we have
\[
\nr{\ad_x^\m \ad_{iA}^m  \tPhil}_{\Lc(H^{s+1},H^s)} \leq C \abs z ^{\nul}
\]
and
\[
 \nr{\ad_x^\m \ad_{iA}^m  \tPhir}_{\Lc(H^{-s},H^{-(s+1)})} \leq C \abs z ^{\nur}.
\]

For $z \in \C_{+}$ and $j,k \in \Ii 0 n$ with $j \leq k$ we set
\begin{equation*} 
\Rc_{j,k} (z) =  R(z) a^{\nu_{j+1}}(x)  R(z)  \dots  a^{\nu_j}(x)  R(z),
\end{equation*}
\begin{equation*} 
\Rc_{j,k}^\y (z) =  \Ri(z) a_\y^{\nu_{j+1}}(x)  \Ri(z)  \dots  a_\y^{\nu_j}(x)  \Ri(z)
\end{equation*}
and
\begin{equation} \label{def-RRiz}
\RRiz_{j,k}(z) =   \Riz(z)  a_{\y,z}^{\nu_{j+1}}(x) \Riz(z) \dots  a_{\y,z}^{\nu_k}(x) \Riz(z).
\end{equation}
For $\a_{j},\dots , \a_{k} \in \N$ we also define
\begin{equation*} %\label{def-Th}
 \tThiota_{j ; \a_{j},\dots , \a_{k} }(z) =  \big( \tPii + 1\big)^{-\a_{j}}  a_{\y,z}^{\nu_{j+1}}(x)  \big( \tPii + 1\big)^{-\a_{j+1}}  \dots a_{\y,z}^{\nu_{k}}(x)  \big( \tPii + 1\big)^{-\a_{k}} .
\end{equation*}
Finally, for all $j,k \in \N$ with $j\leq k$ we set 
\begin{equation*}%\label{matcalV}
\Vc_{j,k} = \sum_{l=j+1}^{k} \n_l. 
\end{equation*}

With the resolvent identity we can check by induction on $m \in \N$ that $\tPhil \RRiz_{0,n}(z) \tPhir$ can be written as a sum of terms either of the form 
\begin{equation} \label{terme1p}
\big( 1 + \hat z ^2 \big)^\b  \tPhil  \tThiota_{0 ; \a_0,\dots,\a_n}(z) \tPhir
\end{equation}
with $\b \in \N$ and $\a_0,\dots, \a_n \in \N^*$, or 
\begin{equation} \label{terme2p}
\big( 1 + \hat z ^2 \big)^\b   \tPhil  \tThiota_{0 ; \a_0,\dots,\a_j} (z)  \RRiz_{j,k}(z)  \tThiota_{k ; \a_{k},\dots , \a_n}(z) \tPhir
\end{equation}
where $\b \in \N$,  $j,k\in\Ii 0 n$, $j\leq k$, $\a_0,\dots, \a_{j-1} , \a_{k+1}, \dots , \a_n \in \N^*$, $\a_j,\a_{k} \in\N$, $\sum_{l=0}^j \a_l \geq m$ and $\sum_{l=k}^{n} \a_l \geq m$.\\

The following two results are Propositions 7.10 and 7.11 in \cite{boucletr14}:

\begin{proposition} \label{prop-Th}
Let $s \in \big[0,\frac d 2+1[$, $s^* \in \big]-\frac d 2-1,0\big]$, $m \in \N$ and $\m \in \N^d$.  Let $j,k \in \Ii 0 n$ be such that $j \leq k$. Let $\a_j,\dots , \a_{k} \in \N^*$ be such that $2 \sum_{l=j}^k \a_{l} - \Vc_{j,k}  \geq (s-s^*)$.
Then there exist $\y_0 \in ]0,1]$ and $C \geq 0$ such that for $\y \in ]0,\y_0]$ and $z \in \C_{+}$ we have
\[
 \nr{\ad_x^\m \ad_{A}^m  \tThiota_{j ; \a_{j},\dots , \a_{k} } (z) }_{\Lc(H^{s^*}, H^{s})} \leq C \abs z ^{\Vc_{j,k}}.
\]
Moreover this also holds when $\a_j = 0$ if $s < \frac d 2 -1$, and when $\a_k = 0$ if $s^* > -\frac d 2 +1$.
\end{proposition}

\begin{proposition} \label{prop-Th3}
 Let $\sss_l ,\sss_r \in \big[0,\frac d 2 \big[$, $\d_l > \sss_l$ and $\d_r > \sss_r$. Let $j,k\in\Ii 0 n$ and $\a_j,\dots,\a_k \in \N$ be such that $\sum_{l=j}^k \a_l$ is large enough (say greater than $\max(\d_l + 2 + \s_l,\d_r + 2 + \s_r)$).
\begin{enumerate} [(i)]
\item If $\a_j,\dots,\a_{k-1} \in \N^*$ then there exists $C \geq 0$ such that for all $z \in \C_{+}$ we have 
\[
 \nr{\pppg x ^{-\d_l} e^{iA \ln  \abs z} \tPhil(z)  \tThiota_{j ; \a_j,\dots,\a_k} (z) \pppg { A}^{\d_l}}_{\Lc(L^2)} \leq C \abs z^{\sss_l + \nul + \Vc_{j,k}}.
\]
\item If $\a_{j+1} ,\dots,\a_{k} \in \N^*$ then there exists $C \geq 0$ such that for all $z \in \C_{+}$ we have 
\[
 \nr{\pppg { A}^{\d_r}  \tThiota_{j ; \a_j,\dots,\a_k} (z) \tPhir e^{-iA \ln  \abs z} \pppg x ^{-\d_r} }_{\Lc(L^2)} \leq C \abs z^{\sss_r +  \nur + \Vc_{j,k}}.
\]
\end{enumerate}
\end{proposition}

In the following proposition we give uniform estimates for $\RRiz_{j,k}(z)$. For this we use the Mourre theory. Since $a$ is of short range, the inserted factors in \eqref{terme-dec-RN} or \eqref{def-RRiz} can be seen as inserted weights in the sense of Section \ref{sec-mourre}. Thus with Theorem \ref{th-mourre} we see that a weaker weight than usual is needed on both sides (with $\d > n + \frac 12$ we recover Proposition 7.12 in \cite{boucletr14}).

\begin{proposition} \label{prop-estim-res-amort}
Let $\d > n + \frac 12 - \Vc$. There exists $\y_0 \in ]0,1]$ and $C \geq 0$ such that for $\y \in ]0,\y_0]$ and $j,k \in \Ii 0 n$ with $j \leq k$ we have
\[
\forall z \in \C_+, \quad  \nr{\pppg A^{-\d} \RRiz_{j,k}(z) \pppg A^{-\d}}_{\Lc(L^2)} \leq C \abs z ^{ \Vc_{j,k}}.
\]
\end{proposition}

\begin{proof}
\stepp Let $J = \big] \frac 13 , 3 \big[$ and $I = \big[ \frac 2 3 , 2 \big]$. The result is clear for $z \in \C_+$ such that $\Re(\hat z) \notin I$. As in the proof of Proposition 7.2 in \cite{boucletr14} we can check that if $\y$ is small enough then $A$ is conjugate to $\tHii$ in the sense of Definition \ref{def-mourre} uniformly in $z \in \C_+$ with $\Re(\hat z) \in J$.

\stepp Let $\n = \Vc_{j,k}$ and $j_1\dots,j_\n$ be such that $j+1 \leq j_1 < j_2 < \dots < j_\n \leq  k$ and 
\[
\singl{j_1,\dots,j_\n} = \singl{\th  \in \Ii{j+1} k \st \n_\th  = 1}.
\]
Set $j_0 = j$ and $j_{\n+1} = k+1$.  For $\th  \in \Ii 0 \n$ we set $n_\th  = j_{\th +1} - j_\th  \in \N^*$.
Let $\Phi_\y(z) = \pppg A a_\y \big(\frac x {\abs z} \big)$. We can write
\[
\RRiz_{j,k}(z) = \Riz^{n_0}(z) \pppg A\inv \Phi_\y(z) \Riz^{n_1}(z) \pppg A \inv \Phi_\y(z) \Riz^{n_2}(z) \dots \pppg A\inv \Phi_\y(z) \Riz^{n_\n}(z) .
\]
For $m \in \N$ we have 
\[
\ad_{iA}^m (\Phi_\y(z)) = \pppg A (-x\cdot \nabla )^m a_{\y,z} = \pppg A \big((-x\cdot \nabla)^m(a_\y) \big)_z.
\]
Then for $u \in \Sc$
\begin{eqnarray*}
\lefteqn{\nr{\ad_{iA}^m (\Phi_\y(z))u}_{L^2} \lesssim \nr{\big((-x\cdot \nabla)^m(a_\y) \big)_z u}_{L^2} + \sum_{q = 1 }^d \nr{x_q D_q\big((-x\cdot \nabla)^m(a_\y) \big)_z u}_{L^2}}\\
&& \lesssim \abs z \nr{u}_{H^1} + \sum_{q = 1 }^d \nr{\big(x_q D_q(-x\cdot \nabla)^m(a_\y) \big)_z u}_{L^2} + \abs z \sum_{q = 1 }^d \nr{ \big(x_q(-x\cdot \nabla)^m(a_\y) \big)_z D_q u}_{L^2}\\
&& \lesssim \abs z \nr{u}_{H^1},
\end{eqnarray*}
and hence
\[
\nr{\ad_{iA}^m (\Phi_\y(z))}_{\Lc(H^1,L^2)} \lesssim \abs z.
\]

\stepp  
We set $\Rc_0(z) = \Riz^{n_0}(z)$ and, for $\th \in \Ii 1 \n$, $\Rc_\th (z) = \Phi_\y(z) \Riz^{n_\th }(z)$. For $\th  \in \Ii 1 {\n-1}$ we set $\d_\th  = 1$. We also set $\d_0 = \d_\n = \d$. For $\th  \in \Ii 0 \nu$ we have
\begin{align*}
j_{\th +1} - j_0 - \th  -1
& = \# \singl{ \s \in \Ii {j+1}{j_{\th +1}-1} \st \n_\s = 0}\\
& \leq  \# \singl{ \s \in \Ii {1}{n} \st \n_\s = 0}\\
& \leq  n - \Vc,
\end{align*}
so for $\th  \in \Ii 0 \n$
\[
\sum_{q = 0}^\th  \d_q = \d + \th  > n- \Vc  + \th  + \frac 12 \geq j_{\th +1} - j_0 - \frac 12 = \sum_{q=0}^{\th } n_q - \frac 12.
\]
Similarly,
\[
j_{\n + 1} - j_l - (\n-\th ) - 1 \leq n - \Vc
\]
and hence 
\[
\sum_{q = \th +1}^{\n+1} \d_j = \d  + \n - \th  > n - \Vc - \frac 12 + \n - \th  \leq j_{\n+1} - j_l - \frac 12 = \sum_{q = \th }^{\n} n_q - \frac 12.
\]
Thus we can apply Theorem \ref{th-mourre}, and the conclusion follows. 
\end{proof}

\begin{proposition} \label{prop-low-freq-iota}
There exists $\y_0 \in ]0,1]$ such that the statement of Proposition \ref{prop-low-freq-R} holds uniformly in $\y \in ]0,\y_0]$ if we replace $R(z)$ by $\Ri(z)$.
\end{proposition}

\begin{proof}
We consider $\sss_l \in [0, s_l[$ and $\sss_r \in [0, s_r[$ such that 
\[
\sss_l + \sss_r  = \min \big( S , 2(n+1)-\Vc - \anul - \anur \big).
\]
In particular we have 
\begin{equation} \label{eq-compte-reg}
\sss_l + \anul < \frac d 2 + 1, \quad \sss_r + \anur < \frac d 2 + 1 \qandq  \sss_l + \anul  + \sss_r + \anur \leq 2(n+1) - \Vc.
\end{equation}
Let $p_l = \frac {2d}{d-2  \sss_l}$ and $p_r =  \frac {2d}{d+2 \sss_r}$. Since $\d_l \geq s_l > \sss_l$ and $\d_r \geq s_r > \sss_r$ we can check that $\pppg x^{-\d_l}$ belongs to $\Lc(L^{p_l},L^2)$ and $\pppg x^{-\d_r} \in \Lc(L^2,L^{p_r})$. Since $\sss_l < \frac d 2$ and $\sss_r < \frac d 2$ we also have the Sobolev embeddings $H^{\sss_l} \subset L^{p_l}$ and $L^{p_r} \subset H^{- \sss_r}$.
We write
\begin{eqnarray*}
\lefteqn{\pppg x^{-\d_l} \Phil \Rc_{0,n}^\y(z) \Phir \pppg x^{-\d_r}}\\
&& = \abs z^{-2(n+1)} \pppg x^{-\d_l} e^{iA \ln \abs z} \tPhil \RRiz_{0,n}(z) \tPhir  e^{-iA \ln\abs z} \pppg x^{-\d_r}.
\end{eqnarray*}
We first estimate the contribution of a term of the form \eqref{terme1p}. With \eqref{eq-compte-reg} we can apply Proposition \ref{prop-Th} (with $s = \sss_l+\anul$ and $s^* = - \sss_r - \anur$). With Proposition \ref{prop-A} this gives
\begin{eqnarray*}
\lefteqn{\nr{\pppg x^{-\d_l} e^{iA \ln \abs z} \tPhil \tThiota_{0;\a_0,\dots,\a_n}(z) \tPhir  e^{-iA \ln\abs z} \pppg x^{-\d_r}}_{\Lc(L^2)}}\\
&& \leq \nr{e^{iA \ln \abs z}}_{\Lc(L^{p_l})} \nr{\tPhil}_{\Lc(H^{\sss_l + \anul},H^{\sss_l})} \nr{\tThiota_{0;\a_0,\dots,\a_n}(z)}_{\Lc(H^{-\sss_r - \anur},H^{\sss_l + \anul})} \\
&& \qquad \times  \nr{\tPhir}_{\Lc(H^{-\sss_r},H^{-\sss_r - \anur})} \nr{e^{-iA \ln\abs z}}_{\Lc(L^{p_r})}\\
&& \lesssim \abs{z}^{\sss_l + \anul + \Vc + \anur + \sss_r} .
\end{eqnarray*}
Now we consider the contribution of a term of the form \eqref{terme2p} with $m$ large enough. According to Propositions \ref{prop-Th3} and \ref{prop-estim-res-amort} we have 
\begin{eqnarray*}
\lefteqn{\nr{\pppg x^{-\d_l} e^{iA \ln \abs z} \tPhil \tThiota_{0;\a_0,\dots,\a_j}(z) \RRiz_{j,k}(z)  \tThiota_{k ; \a_{k},\dots , \a_n}(z)  \tPhir  e^{-iA \ln\abs z} \pppg x^{-\d_r}}_{\Lc(L^2)}}\\
&& \leq \nr{\pppg x^{-\d_l} e^{iA \ln \abs z} \tPhil \tThiota_{0;\a_0,\dots,\a_j}(z) \pppg A^{\d_l}}
\nr{\pppg A^{-\d_l} \Rc^\y_{j,k}(z)   \pppg A^{-\d_r} } \\
&& \qquad \times \nr{\pppg A^{\d_r} \tThiota_{k ; \a_{k},\dots , \a_n}(z) \tPhir  e^{-iA \ln\abs z} \pppg x^{-\d_r}}\\
&& \lesssim \abs{z}^{\sss_l + \anul + \Vc_{0,j}} \abs z^{\Vc_{j,k}} \abs z^{\Vc_{k,n} \anur + \sss_r}= \abs{z}^{\sss_l + \anul + \Vc + \anur + \sss_r} .
\end{eqnarray*}
It only remains to remark that $\sss_l + \anul + \Vc + \anur + \sss_r$ equals $2(n+1)$ or $S + \anul + \anur + \Vc$ to conclude.
\end{proof}

\subsection{Proof of the low frequency estimates in the general case}

Now we use Proposition \ref{prop-low-freq-iota} to prove Proposition \ref{prop-low-freq-R}. For this we have to add the contributions of $W$ and $\Pcc$ in the self-adjoint part, and $a\h_\y$ in the dissipative part. Let $\y_0 > 0$ be given by Proposition \ref{prop-low-freq-iota} and $\y \in ]0,\y_0]$. For $z \in \C_+$ we set  
\[
\Kcco = \Pcc + W \qandq \Kcc = \Kcco  - iza\h_\y.
\]
Then we have the resolvent identities
\begin{equation} \label{res-id}
R(z) = \Ri(z) - \Ri(z) \Kcc R(z) = \Ri(z) - R(z) \Kcc \Ri(z).
\end{equation}
We first estimate a single resolvent with a strong weight:

\begin{proposition} \label{prop-low-freq-n0}
Let $\nul,\nur \in \N^d$ be such that $\anul \leq 1$ and $\anur \leq 1$. Let $\s > 2$.
Then there exist a neighborhood $\Uc$ of 0 in $\C$ and $C \geq 0$ such that for all $z \in \Uc \cap \C_+$ we have 
\[
\nr{\pppg x^{-\s} D^\nul R(z) D^\nur \pppg x^{-\s}}_{\Lc(L^2)} \leq C.
\]
\end{proposition}

\begin{proof}
\stepp 
Assume that the proposition is proved when $\nur = 0$. Then if $\anur = 1$ we can write by the resolvent identity 
\begin{align*}
\nr{\pppg x^{-\s} D^\nul R(z) D^\nur \pppg x^{-\s}}
& \leq \nr{\pppg x^{-\s} D^\nul (\Hz - i)\inv  D^\nur \pppg x^{-\s}} \\
& \quad + (z-i) \nr{\pppg x^{-\s} D^\nul R(z)  \pppg x^{-\s}} \nr{\pppg x^{\s} (\Hz-i)\inv D^\nur \pppg x^{-\s}} \\
& \lesssim 1 .
\end{align*}
We have used the result for $\nur = 0$ and pseudo-differential calculus. Thus it is enough to prove the proposition in the case $\nur = 0$.

\stepp Let $\p \in C_0^\infty(\R,[0,1])$ be equal to 1 on a neighborhood of 0. For $\e > 0$ we set
$
\p_\e : \l \mapsto \p \left( \frac \l \e \right).
$
We recall that $R_0(z) = (\Pg-z^2)\inv$. By \eqref{res-id} and a similar resolvent identity between $R(z)$ and $R_0(z)$ we have 
\begin{align*}
R(z)
& = \Ri(z) \p_\e(P) - R(z) \Kcc \Ri(z) \p_\e(P)\\
& \quad + R_0 (z) (1-\p_\e)(P) + i z R(z) a R_0(z) (1-\p_\e)(P).
\end{align*}
Assume that there exists $\e > 0$ such that 
\begin{equation} \label{estim-Se}
\nr{\pppg x^\s \Kcc \Ri(z) \p_\e(P) \pppg x^{-\s}} \leq \frac 13.
\end{equation}
There exists a neighborhood $\Uc_\e$ of 0 such that for $z \in \Uc_\e \cap \C_+$ the operators $\p_\e(P)$ and ${R_0(z) (1-\p_\e)(P)}$ are bounded in $L^{2,\s}$ so according to Proposition \ref{prop-low-freq-iota} there exists $C,C_0 > 0$ such that 
\[
\nr{\pppg x^{-\s} D^\nul R(z)  \pppg x^{-\s} } \leq C + \left( \frac 13 + C_0 \abs z \right) \nr{\pppg x^{-\s} D^\nul R(z)  \pppg x^{-\s} }.
\]
For $z$ small enough this proves the result. It remains to prove \eqref{estim-Se}. 

\stepp Let $\g_0 \in C_0^\infty$, $j,k \in \Ii 1 d$. According to the Hardy inequality we have for $u \in \Sc$ 
\[
\nr{\pppg x^\s \g_0 D_k u}_{L^2} \lesssim \nr{\pppg x \inv D_k u}_{L^2} \lesssim \nr{u}_{\dot H^2}.
\]
We also have 
\[
\nr{\pppg x^\s D_j \g_0 D_k u}_{L^2} \leq \nr{\pppg x^\s (D^\a \g_0) D_k u}_{L^2} + \nr{\pppg x^\s  \g_0 D^\a D_k u}_{L^2} \lesssim \nr{u}_{\dot H^2}.
\]
If $\y$ is small enough we have 
\[
\nr{v}_{\dot H^2} \leq 2 \nr{\Pii v}_{L^2}
\]
for all $v \in \Sc$ (see Remark 7.8 in \cite{boucletr14}), so for $\m > 0$ we have
\begin{eqnarray} \label{estim-Pcc-W}
\lefteqn{\nr{\pppg x^\s (\Pcc + W) \Ri(i\m) \p_\e(P) \pppg x^{-\s} u}_{L^2} \lesssim \nr{\Ri(i\m) \p_\e(P) \pppg x^{-\s} u}_{\dot H^2}}\\
&&\nonumber \lesssim \nr{\Pii \Ri(i\m) \p_\e(P) \pppg x^{-\s} u}_{L^2}\\
&&\nonumber \lesssim \nr{\p_\e(P) \pppg x^{-\s} u}_{L^2} + \m \nr{a_\y \Ri(i\m) \p_\e(P) \pppg x^{-\s} u}_{L^2} + \m^2 \nr{\Ri(i\m) \p_\e(P) \pppg x^{-\s} u}_{L^2}.
\end{eqnarray}
Then we can finish the proof of \eqref{estim-Se} with the ideas of the proof of Proposition 7.15 in \cite{boucletr14}: by a compactness argument the norm of $\p_\e(P) \pppg x^{-\s}$ goes to 0 when $\e$ goes to 0, for the second term we use the quadratic estimates (see \cite[Proposition 3.7]{boucletr14}) and for the last term we use the easy estimate $\nr{\Ri(i\m)} \leq \m^{-2}$. The difference between $\Ri(\t + i\m)$ and $R(i\m)$ can be written as the integral over $s \in [0,\t]$ of $R'(s+i\t)$ so by the estimate of $\Ri'$ given by Proposition \ref{prop-low-freq-iota} we conclude that the norm \eqref{estim-Pcc-W} is as small as we wish if $\e > 0$ and $z \in \C_+$ are small enough.
It remains to estimate
\[
\abs z \pppg x^\s \h_\y a \Ri(z) \p_\e(P) \pppg x^{-\s}.
\]
But with Proposition \ref{prop-low-freq-iota} and the boundedness of $\p_\e(P)$ in $L^{2,\s}$ we see that for $z \in \C_+$ close to 0 the norm of this operator is as small as we wish. This concludes the proof of \eqref{estim-Se} and hence the proof of the proposition.
\end{proof}

Now we can finish the proof of Proposition \ref{prop-low-freq-R}:

\begin{proof}[Proof of Proposition \ref{prop-low-freq-R}]
For $z \in \C_+$, $j,k \in \Ii 0 n$ with $j \leq k$ we set 
\[
\Rc_{j,k}^0(z) = \Rc^\y_{j,k}(z) \qandq \Rc_{j,k}^1 (z) = \Rc_{j,k}(z).
\]
Using the resolvent identities \eqref{res-id} we can prove by induction on $k \in \Ii 0 n$ that $\Rc_{0,k}(z) = \Rc_{0,k}^1(z)$ can be written as a linear combination of terms of the form 
\[
T(z) = \Rc_{0,j_0}^{0}(z) \Kcc \Rc_{j_0,j_1}^{m_1} \Kcc \dots \Rc_{j_{p-2},j_{p-1}}^{m_{p-1}}(z) \Kcc \Rc_{j_{p-1},j_p}^0(z),
\]
where $p \in \N$, $j_p = k$, $m_1,\dots,m_{p-1} \in \{0,1\}$ and for $\th \in \Ii 1 {p-1}$
\[
m_\th = 1 \quad \implies j_{\th-1} = j_\th.
\]
This last property means that $R(z)$ only appears between two factors $\Kcc$ (see \cite{khenissir} for an analogous argument without inserted factors between the resolvents). We use this result with $k = n$. Let $\d = \max(\d_l,\d_r)$. If $p = 0$ then the contribution of $T(z) = \Rci_{0,n}(z)$ is estimated by Proposition \ref{prop-low-freq-iota}. If $p=1$ we can write 
\begin{eqnarray*}
\lefteqn{\nr{\pppg x^{-\d_l} D^{\n_l} \Rci_{0,j_0}(z) \Kcc \Rci_{j_0,n}(z) D^{\n_r} \pppg x^{-\d_r}}}\\
&& \lesssim \nr{\pppg x^{-\d_l} D^{\n_l} \Rci_{0,j_0}(z) \nabla \pppg x^{-\d}} \nr{\pppg x ^{-\d} \nabla \Rci_{j_0,n}(z) D^{\n_r} \pppg x^{-\d_r}} \\
&& \quad + \sum_{j=1}^d \nr{\pppg x^{-\d_l} D^{\n_l} \Rci_{0,j_0}(z) \pppg x^{-\d}} \nr{\pppg x ^{\d} b_j \nabla \Rci_{j_0,n}(z) D^{\n_r} \pppg x^{-\d_r}} \\
&& \quad +  \nr{\pppg x^{-\d_l} D^{\n_l} \Rci_{0,j_0}(z) \pppg x^{-\d}} \nr{z \pppg x ^{-\d}  \Rci_{j_0,n}(z) D^{\n_r} \pppg x^{-\d_r}}.
\end{eqnarray*}
Then we have by Proposition \ref{prop-low-freq-iota}
\[
\nr{\pppg x^{-\d_l} D^{\n_l} \Rci_{0,j_0}(z) \Kcc \Rci_{j_0,n}(z) D^{\n_r} \pppg x^{-\d_r}}\lesssim  \left( 1 + \abs z^{-2(n+1) + s_l + s_r + \Vc_{0,n}+ \anul + \anur} \right).
\]
It is important to notice that if $\anur = 1$ then $b_1 = \dots = b_d = 0$, so that the second term vanishes. Otherwise this estimate would not hold in the case 
$j_0 = n$ and $s_r = \frac {d} 2$.
We proceed similarly for any $p$ if $m_1 = \dots = m_{p-1} = 0$. If $m_l = 1$ for some $j \in \Ii 1 {p-1}$ then we have $j_{l-1} = j_l$ so according to Proposition \ref{prop-low-freq-n0} the corresponding contribution is uniformly bounded. This concludes the proof of Proposition \ref{prop-low-freq-R}.
\end{proof}

\section{Uniform resolvent estimates in the energy space.} \label{sec-res-estim-energy}

In this section we use the estimates on $R(z)$ to prove Theorems \ref{th-inter-freq}, \ref{th-high-freq} and \ref{th-low-freq}. After differentiation in \eqref{res-Ac} we get for all $N \in \N^*$:
\begin{equation} \label{res-Ac-N}
(\Ac-z)^{-N-1} =
\begin{pmatrix} 
R^{(N)}(z) (ia + z) + R^{(N-1)}(z)  &   R^{(N)}(z)\\
R^{(N)}(z) \Ho &  z R^{(N)}(z) + R^{(N-1)}(z)
\end{pmatrix}.
\end{equation}

\subsection{A generalized Hardy inequality}

Our purpose is to prove estimates of the form 
\[
\nr{(\Ac-z)^{-1-N} U_0}_{\Hc^{-\d}} \leq C(\abs z) \nr{U_0}_{\Hc^\d}.
\]
In particular, if $U_0 = (u_0,i u_1)$ then the estimates shoud not depend on the (weighted) $L^2$-norm of $u_0$, but only on the (weighted) $L^2$-norms of its first derivatives. Thus the Hardy inequality can play a crucial role since it somehow turns the multiplication by $\pppg x\inv$ into a differentiation. More generally, when we have a weight better than needed we can use the following result:

\begin{lemma} \label{lem-hardy-gen}
Let $\d \geq 0$ and $\s < \d - 1$. Then there exists $C \geq 0$ such that for $u \in \Sc$ we have 
\[
\nr{\pppg x^\s u}_{L^2} \leq C \nr{\pppg x^\d \nabla u}_{L^2}.
\]
\end{lemma}

\begin{proof}
Let $\h \in C_0^\infty(\R_+,[0,1])$ be equal to 1 on [0,1]. According to the usual Hardy inequality we have 
\[
\nr{\h(\abs x) \pppg x^\s u}_{L^2} \lesssim \nr{\pppg x\inv u}_{L^2} \lesssim \nr{\nabla  u}_{L^2} \leq \nr{\pppg x^\d \nabla u}_{L^2}.
\]
On the other hand
\[
\nr{(1-\h)(\abs x) \pppg x^\s u}_{L^2}^2 \leq  \int_{\o \in S^{d-1}}\int_{r=1}^{+\infty} \pppg r^{2\s} \abs{u(r \o)}^2 r^{d-1}  \, dr \, d\o.
\]
For $r \geq 1$ and $\o \in S^{d-1}$ we have
\begin{align*}
\abs{u(r\o)}^2 
& = \abs{\int_{s= r}^{+\infty} (\o \cdot \nabla) u (s \o) s^{\frac {2\d + d - 1} 2} s^{- \frac {2\d + d -1}2} \, ds   }^2\\
& \leq \int_{s = r}^{+\infty} \abs{\nabla u(s\o)}^2 s^{2\d} s^{d-1} \, ds \times \int_{s=r}^{+\infty} s^{-2\d - d + 1} \, ds\\
& \lesssim r^{-2\d +2 - d} \int_{s = r}^{+\infty} \abs{\nabla u(s\o)}^2 \pppg s^{2\d} s^{d-1} \, ds .
\end{align*}
Thus
\begin{align*}
\nr{(1-\h)(\abs x) \pppg x^\s u}_{L^2}^2
&  \lesssim \int_{r= 1}^{+\infty} \pppg r^{2(\s +1 - \d) -1}   \int_{\o \in S^{d-1}}\int_{s = r}^{+\infty} \abs{\nabla u(s\o)}^2 \pppg s^{2\d} s^{d-1} \, ds \, d\o\, dr\\
& \lesssim \nr{\pppg x^\d \nabla u}_{L^2}^2.
\end{align*}
The lemma is proved.
\end{proof}

\subsection{Intermediate and high frequency estimates}

We begin with the proofs of Theorems \ref{th-inter-freq} and \ref{th-high-freq}. The proofs are exactly the same. The only difference is that we need the damping assumption on bounded geodesics to use the estimates of Proposition \ref{prop-high-freq} uniformly in $\abs z \gg 1$.

\begin{proof} [Proof of Theorems \ref{th-inter-freq} and \ref{th-high-freq}]
Let $\g > 0$. Let $(u,v) \in \Sc\times \Sc$. In this proof the symbol $\lesssim$ stands for $\leq C$ for a constant $C$ which does not depend on $u$, $v$ and $z \in \C_+$ with $\g \leq \abs z \leq \g \inv$ (or $\abs z \geq \g$ under the damping assumption on bounded geodesics). According to Proposition \ref{prop-high-freq} we have
\[
 \nr{\pppg x^{-\d} \nabla  R^{(N)}(z) v}  \lesssim  \nr{\pppg x^\d v}_{L^2}
\]
and
\[
\abs z \nr{\pppg x^{-\d} R^{(N)}(z) v} + \nr{\pppg x^{-\d} R^{(N-1)}(z) v}_{L^2} \lesssim  \nr{\pppg x^\d v}_{L^2}.
\]
We now consider the contribution of the lower left coefficent in \eqref{res-Ac-N}. We have 
\begin{equation} \label{lower-left}
\begin{aligned}
\nr{\pppg x^{-\d}  R^{(N)}(z) \Ho u}_{L^2} 
& \lesssim \sum_{j,k = 1}^d \nr{\pppg x^{-\d} R^{(N)}(z) D_j \pppg x^{-\d}}_{\Lc(L^2)} \nr{\pppg x^\d G_{j,k}(x) D_k u}_{L^2}\\
& \quad + \sum_{k = 1}^d \nr{\pppg x^{-\d} R^{(N)}(z)  \pppg x^{-\d}}_{\Lc(L^2)} \nr{\pppg x^\d b_k(x) D_k u}_{L^2} \\
& \lesssim  \nr{\pppg x^\d \nabla u}_{L^2}.
\end{aligned}
\end{equation}
For the upper left coefficent in \eqref{res-Ac} we write 
\[
R(z) (ia + z) = \frac 1 z R(z)\Big( \Ho - \big(\Ho -iaz - z^2\big) \Big) = \frac 1 z R(z) \Ho - \frac 1 z.
\]
Thus
\begin{equation} \label{der-N-high-freq}
\frac {d^N}{dz^N} \big( R(z) (ia + z) \big)  = \sum_{l=0}^N \frac {C_N^l (-1)^{l} l!}{z^{1+l}} R^{(N-l)}(z) \Ho + \frac {(-1)^{N+1} N!}{ z^{1+N}},
\end{equation}
and hence
\begin{eqnarray*}
\lefteqn{\nr{\pppg x^{-\d} \nabla \frac {d^N}{dz^N} \big( R(z) (ia + z) \big) u}} \\
&& \lesssim \sum_{j,k=0}^d \sum_{l=0}^N \abs z^{-1-l} \nr{\pppg x^{-\d} \nabla R^{(N-l)}(z) D_j \pppg x^{-\d}} \nr{\pppg x^\d D_k u} \\
&& \quad + \sum_{k=0}^d \sum_{l=0}^N \abs z^{-1-l} \nr{\pppg x^{-\d} \nabla R^{(N-l)}(z) \pppg x^{-\d}} \nr{\pppg x^\d D_k u} + \abs z^{-1-N} \nr{\pppg x^{-\d} \nabla u}  \\
&& \lesssim \nr{\pppg x^\d \nabla u}.
\end{eqnarray*}
This concludes the proof.
\end{proof}

\subsection{Low frequency estimates}

We now turn to the proof of Theorem \ref{th-low-freq} concerning low frequencies. For the coefficents on the right in \eqref{res-Ac} the proof is as simple as for high frequencies: for $(u,v) \in \Sc \times \Sc$ we have
\begin{equation} \label{estim-upper-right}
 \nr{\pppg x^{-\d} \nabla  R^{(N)}(z) v}  \lesssim  \left( 1 + \abs z^{d-N- 1 -\e} \right)\nr{\pppg x^\d v}_{L^2}
\end{equation}
and
\begin{equation} \label{estim-lower-right}
\abs z \nr{\pppg x^{-\d} R^{(N)}(z) v} + \nr{\pppg x^{-\d} R^{(N-1)}(z) v}_{L^2} \lesssim  \left(1 + \abs z^{d-N- 1 -\e} \right)\nr{\pppg x^\d v}_{L^2}.
\end{equation}

This is not the case for the coefficents on the left. We have to be careful with the lack of derivative in the contribution of $W$ in $\Pg$ (see \eqref{lower-left} for the lower left coefficient). Moreover, we cannot follow the proof of Theorems \ref{th-inter-freq} and \ref{th-high-freq} for the upper left coefficient when $z$ is small because of the powers of $z \inv$ in \eqref{der-N-high-freq}.\\

Let $\Pgg = -\divg G(x) \nabla$ and for $z \in \C_+$: $\Rgg(z) = \big( \Pgg -iza(x) -z^2 \big)\inv$. We have 
\[
\Ac = \Agg + \Wc
\]
where 
\[
\Agg = \begin{pmatrix} 0 & 1 \\ \Pgg & -ia \end{pmatrix} \qandq  \Wc = \sum_{j=1}^d \begin{pmatrix} 0 & 0 \\ b_j D_j & 0 \end{pmatrix} \in \Lc(\Hc).
\]

\begin{proposition} \label{prop-Agg-Ac}
Assume that the estimates of Theorem \ref{th-low-freq} holds with $(\Ac-z)\inv$ replaced by $(\Agg-z)\inv$. Then the same estimates hold for $(\Ac-z)\inv$.
\end{proposition}

\begin{proof}
For $z\in \C_+$ we have the resolvent identities
\begin{align*}
(\Ac -z)\inv
&= (\Agg -z)\inv -  (\Agg -z)\inv \Wc (\Ac -z)\inv\\ 
&= (\Agg -z)\inv -  (\Ac -z)\inv \Wc (\Agg -z)\inv.
\end{align*}
Then we can check that
\begin{align*}
(\Ac-z)^{-1-N}
& = (\Agg-z)^{-1-N} - \sum_{k=0}^{N} (\Agg-z)^{-1-k} \Wc (\Agg-z)^{-1-N+k}\\
& + \sum_{j+k \leq N} (\Agg-z)^{-1-k} \Wc (\Ac-z)^{-1-N+k +j} \Wc (\Agg-z)^{-1-j}.
\end{align*}
Let $j,k \in \N$ be such that $j + k \leq N$. For $(u,v) \in \Sc^2$ we have 
\[
\Wc (\Ac-z)^{-1-N+k +j} \Wc  \begin{pmatrix} u \\ v \end{pmatrix} = \sum_{m,l = 1}^d \begin{pmatrix} 0 \\ b_m D_m R^{(N-k-j)}(z) b_l D_l u \end{pmatrix}
\]
In particular
\[
\nr{\Wc (\Ac-z)^{-1-N+k +j} \Wc}_{\Lc(\Hc^{-\d},\Hc^\d)} \lesssim 1 + \abs {z}^{d-1-N +j +k -\e}.
\]
With the estimates on the powers of $(\Agg-z)\inv$ we obtain the estimates of Theorem \ref{th-low-freq}.
\end{proof}

Thus, Theorem \ref{th-low-freq} is a consequence of the following result:

\begin{proposition} \label{prop-res-Agg}
The estimates of Theorem \ref{th-low-freq} hold with $(\Ac-z)\inv$ replaced by $(\Agg-z)\inv$.
\end{proposition}

\begin{proof} 
\stepp Estimates \eqref{estim-upper-right} and \eqref{estim-lower-right} hold for $\Rgg$. For the lower left coefficient we can use the estimates of Theorem \ref{th-low-freq-R}. For $u \in \Sc$ we have 
\begin{align*}
\nr{\pppg x^{-\d}  \Rgg^{(N)}(z) \Pgg u}_{L^2}
& \lesssim \sum_{j,k = 1}^d \nr{\pppg x^{-\d} \Rgg^{(N)}(z) D_j \pppg x^{-\d}}_{\Lc(L^2)} \nr{\pppg x^\d D_k u}_{L^2}\\
& \lesssim  \left(1 + \abs z^{d-N- 1 -\e} \right) \nr{\pppg x^\d \nabla u}_{L^2}.
\end{align*}

\stepp We now estimate the upper left coefficient for $N \geq 1$. Let 
\[
\s \in  \left] N - \frac 12 , \d-1 \right[.
\]
We estimate separately the two terms which appear in \eqref{res-Ac-N}. For the second term we use Lemma \ref{lem-hardy-gen}:
\begin{align*}
\nr{\pppg x^{-\d}   \nabla \Rgg^{(N-1)}(z)  u}_{L^2}
& \lesssim \nr{\pppg x^{-\s} \nabla \Rgg^{(N-1)}(z) \pppg x^{-\s}}_{\Lc(L^2)} \nr{\pppg x^\s u}_{L^2}\\
& \lesssim \left(1 + \abs z^{d-(N-1)- 1 -\e} \right) \nr{\pppg x^\d \nabla u}_{\Lc(L^2)}.
\end{align*}
We now turn to the contribution of $\Rgg^{(N)}(z) (ia+z)$. Let $j\in\Ii 1 d$ and $T(z)$ be a term of the decomposition of $\Rgg^{(N)}(z)$ given by Proposition \ref{prop-der-R2}. 
If $n < N$ then $\s > n + \frac 12$, so by Proposition \ref{prop-low-freq-R} we have
\[
\nr{\pppg x^{-\s} D_j T(z) (ia+z) u} \lesssim \left(1 + \abs z^{d- N - 1 -\e} \right) \nr{\pppg x^\s u},
\]
and we can conclude with Lemma \ref{lem-hardy-gen}.
If $n = N$ and $\Vc = 0$ then we have $\o = N$. In this case we can write 
\begin{align*}
\nr{\pppg x^{-\d} D_j T(z) (ia+z) u}
& = \abs z^{N} \nr{\pppg x^{-\d} D_j \Rgg^{N+1} (z) (ia + z) u}\\
& = \abs{z}^{N-1} \nr{\pppg x^{-\d} D_j \Rgg^{N+1} (z) \big( \Pgg - (\Pgg -iza(x) - z^2) \big) u}\\
& \lesssim \sum_{l,k} \abs{z}^{N-1} \nr{\pppg x^{-\d} D_j \Rgg^{N+1} (z) D_l \pppg x^{-\d}} \nr{\pppg x^\d D_k u}\\
& \qquad  + \abs{z}^{N-1} \nr{\pppg x^\d D_j \Rgg^{N}(z) \pppg x^{-\s}} \nr{\pppg x^\s u}\\
& \lesssim   \left(1 + \abs z^{d-N- 1 -\e} \right) \nr{\pppg x^\d \nabla u}_{L^2}.
\end{align*}
Using the decay of $a$ we have in any case
\[
\nr{\pppg x^{-\d} D_j T(z) a u} \lesssim \left(1 + \abs z^{d- N - 1 -\e} \right) \nr{\pppg x^\s u} \lesssim  \left(1 + \abs z^{d- N - 1 -\e} \right) \nr{\pppg x^\d \nabla u},
\]
so it remains to estimate $z \pppg x^{-\d} D_j T(z) u$ when $n = N$ and $\Vc \geq 1$. For this use Corollary \ref{cor-sigma-r-un}:
\begin{equation} \label{estim-sigma-r-un}
\abs z \nr{\pppg x^{-\s} D_j T(z) u} \lesssim \abs z \left( 1 +\abs z^{N-2} + \abs z ^{d-N-1-\e} \right) \nr{\pppg x^\s u}.
\end{equation}
We conclude again with Lemma \ref{lem-hardy-gen}.

\stepp We now estimate the contribution of the upper left coefficient in \eqref{res-Ac} when $N = 0$. More precisely we have to prove that for $\d > \frac 12$ we have
\begin{equation} \label{estim-N-nul}
\nr{\pppg x^{-\d} D_j \Rgg(z) (ia+z)u} \lesssim \nr{\pppg x^\d \nabla u}.
\end{equation}
Without loss of generality we can assume that $\d < \frac 12 + \frac \rho 4$. Then by Lemma \ref{lem-hardy-gen} we have
\begin{equation} \label{estim-DjRGa}
\nr{\pppg x^{-\d} D_j \Rgg (z) a  u} \lesssim \nr{\pppg x^{-\d} D_j \Rgg(z) \pppg x^{-\d}} \nr{\pppg x^{-\d}  u} \lesssim \nr{\pppg x^\d \nabla u} .
\end{equation}
It remains to estimate $\pppg x^{-\d} D_j \Rgg (z) z  u$. For $\y \in ]0,1]$ we set $R^0_G(z) = (\Pgg - z^2)\inv$ and $R_{0,\y}(z) = (\Pii - z^2)\inv$, where $\Pii$ is as in \eqref{def-Pii-Pcc}. We have 
\[
[D_j, R_{0,\y}(z)] = \sum_{k,l=1}^d  R_{0,\y}(z) D_k \big(D_j G_{\y,k,l}(x) \big) D_l R_{0,\y}(z)
\]
so
\begin{eqnarray*}
\lefteqn{\nr{\pppg x^{-\d} D_j R_{0,\y}(z)z u}}\\
&& \leq \nr{\pppg x^{-\d} R_{0,\y}(z) z D_j  u} +  \sum_{k,l=1}^d \nr{\pppg x^{-\d} R_{0,\y}(z) D_k \big(D_j G_{\y,k,l}(x) \big) D_l R_{0,\y}(z) z u} \\
&& \lesssim \nr{\pppg x^{-\d} R_{0,\y}(z) z \pppg x^{-\d}} \nr{\pppg x^\d  D_j  u}\\
&& \qquad  + \y^{\frac \rho 2} \sup_{k,l} \nr{\pppg x^{-\d} R_{0,\y} D_k \pppg x^{-\frac 12 - \frac \rho 4} } \nr{\pppg x^{- \frac 12 - \frac \rho 4} D_l R_{0,\y}(z) z u}
\end{eqnarray*}
Then according to the second statement in Theorem \ref{th-low-freq-R} (applied with $\s = 0$ and $a= 0$) there exists $C \geq 0$ such that for $\y > 0$ small enough we have 
\[
\sum_{j} \nr{\pppg x^{-\d} D_j R_{0,\y}(z)z u} \leq C \nr{\pppg x^\d \nabla u} + \frac 12 \sum_{j} \nr{\pppg x^{-\d} D_j R_{0,\y}(z)z u},
\]
and hence 
\[
\sum_{j} \nr{\pppg x^{-\d} D_j R_{0,\y}(z)z u} \lesssim \nr{\pppg x^\d \nabla u}.
\]
Now $\y > 0$ is fixed. Let $G_0(x) = \h_\y(x) (G(x) - I_d)$. We have the resolvent identity
\[
R^0_G(z)(z) = R_{0,\y}(z) + \sum_{k,l} R^0_G(z)(z) D_k G_{0,k,l}(x) D_l R_{0,\y}(z),
\]
so 
\begin{eqnarray*}
\lefteqn{\nr{\pppg x^{-\d} D_j R^0_G(z)(z)z u}}\\
&& \leq \nr{\pppg x^{-\d} D_j R_{0,\y}(z)z u} + \sum_{k,l} \nr{\pppg x^{-\d} D_j R^0_G(z)(z) D_k G_{0,k,l}(x) D_l R_{\y,0}(z)z u}\\
&& \lesssim \nr{\pppg x^\d \nabla u} +  \nr{\pppg x^{-\d} D_j R^0_G(z)(z) D_k \pppg x^{-\d}}\nr{\pppg x^{-\d} D_l R_{\y,0}(z)z u}\\
&& \lesssim \nr{\pppg x^\d \nabla u}.
\end{eqnarray*}
It remains to add the dissipative part. Again, we use the corresponding resolvent identity
\[
\Rgg(z) = R^0_G(z)(z) + iz \Rgg(z) a(x) R^0_G(z)(z)
\]
to write 
\begin{eqnarray*}
\lefteqn{\nr{\pppg x^{-\d} D_j \Rgg(z) z u} \leq \nr{\pppg x^{-\d} D_j R^0_G(z)(z)z u} + \nr{\pppg x^{-\d} D_j \Rgg(z) a R^0_G(z)(z) z^2 u}}\\
&& \lesssim \nr{\pppg x^\d \nabla u} + \nr{\pppg x^{-\d} D_j \Rgg(z) a R^0_G(z)(z) \big( \Pgg - (\Pgg-z^2)\big) u}\\
&& \lesssim \nr{\pppg x^\d \nabla u} + \sum_{k,l} \nr{\pppg x^{-\d} D_j \Rgg(z) \pppg x^{-\d}} \nr{\pppg x^{-\d} R^0_G(z)(z) D_k \pppg x^{-\d} } \nr{\pppg x^\d D_l u}\\
&& \qquad  + \nr{\pppg x^{-\d} D_j \Rgg(z) a u}\\
&& \lesssim \nr{\pppg x^\d \nabla u} .
\end{eqnarray*}
We have used \eqref{estim-DjRGa}. This concludes the proof of \eqref{estim-N-nul} and hence the proof of the proposition.
\end{proof}

\section{Local energy decay in the energy space.} \label{sec-loc-decay}

In this section we use the resolvent estimates of Theorems \ref{th-inter-freq}, \ref{th-high-freq} and \ref{th-low-freq} to prove Theorem \ref{th-loc-decay-improved}. We begin with a lemma about the propagation for finite times. It relies on the propagation at finite speed for the wave equation:

\begin{lemma} \label{lem-tps-fini}
Let $\d \geq 0$ and $T > 0$. Then there exists $C_T \geq 0$ such that for all $t \in [0,T]$ and $U_0 \in \Sc \times \Sc$ we have 
\[
\nr{e^{-it\Ac}U_0}_{\Hc^\d} \leq C_T \nr{U_0}_{\Hc^\d}.
\]
\end{lemma}

\begin{proof}
For $\d, r_1,r_2 \in \R$ and $u,v \in \Sc$ we set 
\[
\nr{(u,v)}_{\Hc^\d(r_1,r_2)}^2 = \int_{r_1 \leq \abs x \leq r_2} \pppg x^{2\d} \big( \abs{\nabla u}^2 + \abs v^2 \big) \, dx.
\]
We also write $\Hc(r_1,r_2)$ for $\Hc^0(r_1,r_2)$. Let $u$ be a solution of \eqref{wave-lap} and $U : t \mapsto (u(t),i\partial_t u(t))$. Let $r_1,r_2 \in \R_+$ with $r_1 \leq r_2$. For $t \geq 0$ and $s \in [0,t]$ we have
\begin{eqnarray*}
\lefteqn{\frac d {ds} \nr{U(t-s)}_{\Hc(r_1 - s,r_2 + s)}^2}\\
&& = \int_{\abs x = r_2+s} \big( \abs {\nabla u(t-s)}^ 2 + \abs {\partial_t u(t-s)}^2 \big) + \int_{\abs x = r_1 -s} \big( \abs {\nabla u(t-s)}^ 2 + \abs {\partial_t u(t-s)}^2 \big) \\
% \frac d {ds} \int_{r=a-s}^{b+s} E_S(t-s,r) \, dr}\\
% && \\
% && = E_S(t-s,b+s) + E_S(t-s,a-s)
&& \quad - 2 \Re\int_{r_1-s \leq \abs x \leq r_2 + s} \big(  \nabla u(t-s) \cdot \partial_t\nabla \bar u(t-s) + \partial_t^2  u (t-s) \, \partial_t \bar u(t-s) \big)\\
&& = \int_{\abs x = r_2 + s} \big(\abs {\nabla u(t-s)}^ 2 + \abs {\partial_t u(t-s)}^2  -  2 \Re \partial_r u(t-s) \, \partial _t \bar u(t-s)  \big)\\
&& \quad + \int_{\abs x = r_1 - s} \big(\abs {\nabla u(t-s)}^ 2 + \abs {\partial_t u(t-s)}^2  +  2 \Re \partial_r u(t-s) \, \partial _t \bar u(t-s)  \big)\\
&& \quad + 2 \int_{r_1 - s \leq \abs x \leq  r_2 + s} a \abs{\partial_t u(t-s)}^2 \\
&& \geq 0.
\end{eqnarray*}
We have denoted by $\partial_r u$ the radial derivative of $u$ with respect to the spacial variable. Moreover the integrals over $\singl{\abs x = r_1-s}$ vanish when $s \geq r_1$.
This proves that 
\[
\nr{U(t)}_{\Hc(r_1,r_2)} \leq \nr{U(0)}_{\Hc(r_1-t,r_2 + t)}.
\]

\stepp 
Let $U_0 \in \Sc \times \Sc$. For $t \in [0,T]$ we have 
\begin{align*}
\nr{e^{-it\Ac}U_0}^2_{\Hc^\d} 
& \leq  \pppg T ^ {2\d} \nr{e^{-it\Ac}U_0}^2_{\Hc(0,T)} + \sum_{n \in \N} \pppg {n+T+1}^{2\d} \nr{e^{-it\Ac}U_0}^2_{\Hc(T+n,T+n+1)}\\
& \leq  \pppg T ^ {2\d} \nr{U_0}^2_{\Hc(0,2T)} + \sum_{n \in \N} \pppg {n+T+1}^{2\d} \nr{U_0}^2_{\Hc(n,2T+n+1)}\\
& \leq  \pppg T ^ {2\d} \nr{U_0}^2_{\Hc^\d} + \sum_{n \in \N} \frac {\pppg {n+T+1}^{2\d}}{\pppg n^{2\d}} \nr{U_0}^2_{\Hc^\d(n,2T+n+1)}\\
& \leq  \pppg T ^ {2\d} \nr{U_0}^2_{\Hc^\d} + (2T+2) \nr{U_0}^2_{\Hc^\d} \sup_{n\in\N} \frac {\pppg {n+T+1}^{2\d}}{\pppg n^{2\d}} \\
& \lesssim \nr{U_0}_{\Hc^\d}.
\end{align*}
This concludes the proof.
\end{proof}

Now we can prove Theorem \ref{th-loc-decay-improved}:

\begin{proof} [Proof of Theorem \ref{th-loc-decay-improved}]
Let $U_0 \in \Sc \times \Sc$. We denote by $U(t)$ the solution of \eqref{wave-A}. Let $\h \in C^\infty(\R, [0,1])$ be equal to 0 on $]-\infty,1[$ and equal to 1 on $]2,+\infty[$.

\stepp Let $z \in \C_+$. We multiply \eqref{wave-A} by $e^{itz} \h(t)$ and take the integral over $\R$. After a partial integration we get
\begin{equation} \label{eq-Fourier-U-V}
\int_\R \h(t) e^{itz} U(t) \, dt = (\Ac-z)\inv V(z), 
\end{equation}
where we have set 
\[
V(z) = -i \int_\R \h'(t) e^{itz} U(t) \, dt = -i \int_1^2 \h'(t) e^{itz} U(t) \, dt.
\]

\stepp Let $\m > 0$. The map $t \mapsto e^{-t\m} \h(t) U(t)$ and its inverse Fourier transform 
\[
\t \mapsto {(\Ac - (\t+i\m))\inv} {V(\t +i\m)}
\]
belong to $\Sc(\R)$ so we can inverse \eqref{eq-Fourier-U-V}:
\[
\h(t) e^{-t\m} U(t) = \frac 1 {2\pi} \int_\R e^{-it\t}  \big(\Ac-(\t+i\m)\big) \inv V(\t + i\m) \, d\t.
\]
Let $\h_0 \in C_0^\infty(\R,[0,1])$ be equal to 1 on a neighborhood of 0. Let $\h_1 = 1-\h_0$. We can write 
\begin{equation} \label{dec-U0U1}
\h(t) e^{-t\m} U(t) = \frac 1 {2\pi} \big( U_{0,\m} (t) + U_{1,\m} (t) \big)
\end{equation}
where for $j \in \{0,1\}$ we have set
\[
U_{j,\m} (t) = \int_\R \h_j(\t) e^{-it\t} \big(\Ac-(\t+i\m)\big) \inv V(\t + i\m) \, d\t.
\]

\stepp 
With partial integrations we see that 
\[
(it)^{d-1} U_{0,\m}(t) =  \int_\R e^{-it\t} f_\m(\t) \, d\t
\]
where
\[
f_\m(\t) = \left(\frac d {d\t} \right)^{d-1} \left( \h_0(\t) \big(\Ac-(\t+i\m)\big) \inv V(\t + i\m) \right).
\]
According to Lemma \ref{lem-tps-fini} we have for any $k \in \N$
\begin{equation*} %\label{estim-V}
\nr{V^{(k)}(z)}_{\Hc^\d} \lesssim \int_1^2 \nr{U(t)}_{\Hc^\d} \lesssim \nr{U_0}_{\Hc^\d}.
\end{equation*}
With Theorem \ref{th-inter-freq} and Theorem \ref{th-low-freq} applied with $\frac \e 2$ we obtain that there exists $C \geq 0$ such that for all $\t \in \R^*$ and $\m > 0$ we have 
\[
\nr{f_\m(\t)}_{\Hc^{-\d}} \leq C \abs \t^{-\frac \e 2} \nr{U_0}_{\Hc^\d} \quad \text{and} \quad \nr{f'_\m(\t)}_{\Hc^{-\d}} \leq C \abs \t^{-1- \frac \e 2} \nr{U_0}_{\Hc^\d}.
\]
According to Lemma 4.3 in \cite{boucletr14} we obtain
\begin{equation} \label{estim-U0}
\nr{U_{0,\m}(t)}_{\Hc^{-\d}} \lesssim \pppg t^{-(d -\e)} \nr{U_0}_{\Hc^\d}.
\end{equation}

\stepp We now estimate $U_{1,\m}(t)$. Let $N \in \N$ and $\d_N \in \big] N + \frac 12,N+1\big[$. As above we see that $(it)^N U_{1,\m}(t)$ is a linear combination of terms of the form 
\[
\tilde U _{\m,j,k,l}(t) :=  \int_\R  e^{-it\t}  \h_1^{(l)}(\t)\big(\Ac-(\t+i\m)\big)^{-j-1} V^{(k)}(\t + i\m)  \, d\t,
\]
where $j,k,l \in \N$ are such that $j+k+l = N$. According to the Plancherel Theorem (used twice), Theorems \ref{th-inter-freq}, \ref{th-high-freq} and Lemma \ref{lem-tps-fini} we have 
\begin{equation} \label{estim-U-L2}
\begin{aligned}
\int_{\R} \nr {\tilde U _{\m,j,k,l}(t)}_{\Hc^{-{\d_N}}}^2 \, dt
& = \int_{\R} \nr {\h_1^{(l)}(\t) (\Ac-(\t+i\m))^{-j-1} V^{(k)}(\t+i\m)}_{\Hc^{-\d_N}}^2 \, d\t\\
& \lesssim  \int_{\R} \nr {V^{(k)}(\t+i\m)}_{\Hc^{\d_N}}^2 \, d\t\\
& \lesssim \nr{U_0}_{\Hc^{\d_N}}^2.
\end{aligned}
\end{equation}
In particular there exists $C \geq 0$ such that for $U_0 \in \Sc \times \Sc$ we can find $T(U_0) \in [0,1]$ which satisfies
\[
\nr{\tilde U _{\m,j,k,l}(T(U_0))}_{\Hc^{-\d_N}} \leq C \nr{U_0}_{\Hc^{\d_N}}.
\]
For $t \geq 1$ we have
\begin{equation} \label{eq-tilde-U}
\tilde U _{\m,j,k,l}(t) = e^{-i(t-T(U_0))\Ac} \tilde U _{\m,j,k,l}(T(U_0)) +  \int_{T(U_0)} ^t \frac {\partial}{\partial s} \left( e^{-i(t-s)\Ac} \tilde U _{\m,j,k,l}(s) \right) \,ds,
\end{equation}
where for $s \in [T(U_0),t]$
\begin{eqnarray*}
\lefteqn{\frac {\partial}{\partial s} \left( e^{-i(t-s)\Ac}  \tilde U _{\m,j,k,l}(s) \right)}\\
&& = \frac {\partial}{\partial s}\int_\R \h_1^{(l)}(\t) e^{-is\t} e^{-i(t-s)\Ac} \big(\Ac -(\t+i\m)\big)^{-j-1} V^{(k)}(\t+i\m) \, d\t\\
&& = i \int_\R \h_1^{(l)}(\t) e^{-is\t} e^{-i(t-s)\Ac} (\Ac-\t) \big(\Ac -(\t+i\m)\big)^{-j-1} V^{(k)}(\t+i\m)  \, d\t\\
&& = -\m e^{-i(t-s)\Ac}   \tilde U _{\m,j,k,l}(s)  + i e^{-i(t-s)\Ac}  \int_\R e^{-is\t} \h_1^{(l)}(\t) \big(\Ac -(\t+i\m)\big)^{-j} V^{(k)}(\t+i\m) \, d\t.
\end{eqnarray*}
With \eqref{estim-U-L2} and a similar computation for the second term this proves that the map $s \mapsto \frac {\partial}{\partial s} \left( e^{-i(t-s)\Ac}  \tilde U _{\m,j,k,l}(s) \right)$ belongs to $L^2([0,t],\Hc^{-\d_N})$ uniformly in $t$, and its $L^2([0,t], \Hc^{-\d_N})$ norm is controlled by the norm of $U_0$ in $\Hc^{\d_N}$. Using the Cauchy-Schwarz inequality in \eqref{eq-tilde-U} we obtain
\begin{equation*} 
\nr{\tilde U _{\m,j,k,l}(t)}_{\Hc^{-\d_N}} \lesssim (1+ \sqrt t) \nr{U_0}_{\Hc^{\d_N}},
\end{equation*}
and hence 
\begin{equation} \label{estim-HN}
\nr{U_{1,\m}(t)}_{\Hc^{-\d_N}} \lesssim \pppg t^{-N + \frac 12} \nr{U_0}_{\Hc^{\d_N}}.
\end{equation}

\stepp Let $\g \geq 0$ and $\s > \g + \frac 12$. Let $\d_0 \in \big] \frac 12 , \s - \g\big[$. According to \eqref{dec-U0U1} we have
\begin{equation} \label{estim-H0}
\nr{U_{1,\m}(t)}_{\Hc^{-\d_0}} \leq 2\pi \nr{\h(t) e^{-t\m} U(t)}_{\Hc^{-\d_0}} + \nr{U_{0,\m}(t)}_{\Hc^{-\d_0}}  \lesssim  \nr{U_0}_{\Hc^{\d_0}}.
\end{equation}
Now we use interpolation between \eqref{estim-H0} and \eqref{estim-HN} for large $N$. Let $\th_N = {\g} \left( N - \frac 12\right) \inv$. For $N$ large enough we have
\[
\th_N \d_N + (1-\th_N) \d_0 \leq \frac {\g}{ N - \frac 12} (N+1) + \d_0 < \s,
\]
so by interpolation we get 
\begin{equation} \label{estim-U1}
\nr{U_{1,\m}(t)}_{\Hc^{-\s}} \lesssim \pppg t^{-\g} \nr{U_0}_{\Hc^{\s}}.
\end{equation}
In particular this can be applied with $\g = d$ and $\s = \d$. With \eqref{estim-U0}, this concludes the proof.
\end{proof}

\begin{remark}
As usual, the restriction of the time decay in Theorem \ref{th-loc-decay-improved} is due to the contribution of low frequencies. According to \eqref{estim-U1}, the contribution of high frequencies decays like any power of $t$, as long as we choose a suitable weight. This is close to the result of \cite{wang87} (except that we have $\s > \g + \frac 12$ instead of $\s = \g$), even if in our dissipative context we do not have a functionnal calculus to localize the solution spectrally on high frequencies.
\end{remark}

\begin{paragraph}{\bf Acknowledgements}
I am very grateful to Jean-Marc Bouclet for fruitful discussions about this paper. This work is partially supported by the french ANR Project NOSEVOL (ANR 2011 BS01019 01).
\end{paragraph}

\bibliographystyle{alpha}
\bibliography{bibliotex}

\end{document}